\documentclass[perreview,12pt]{article}

\newif \ifshowcomment
\showcommenttrue
\newif \ifshowauthor
\showauthortrue

\usepackage{listings}

\usepackage{comment}

\usepackage{amsmath}
\usepackage{amsfonts}

\usepackage{mathrsfs}

\usepackage{tikz, xstring, calc, pgfmath}

\usepackage{multirow}

\usepackage{ifthen}

\usepackage{url}
\newlength{\minwidth}

\newcommand{\emptytrace}{\epsilon}
\newcommand{\trace}{\tau}
\newcommand{\traceproj}{\bar{\tau}}

\newcommand{\astate}{q}

\newcommand{\ruleName}[1]{\ensuremath{\textsc{#1}}}

\newcommand{\signature}{\Sigma}

\newcommand{\grammarequal}{::=}
\newcommand{\msgset}{\z{Msg}}
\newcommand{\constset}{\z{Const}}
\newcommand{\nonceset}{\z{Nonce}}
\newcommand{\agentset}{\z{Agent}}
\newcommand{\honestset}{\z{Honest}}
\newcommand{\dishonestset}{\z{Dishonest}}

\newcommand{\privfunset}{\{ \shakey{}{}\}}

\newcommand{\anonce}{n}

\newcommand{\nonce}[1]{N_{#1}}

\newcommand{\nonceP}{\nonce{\prover}}
\newcommand{\nonceV}{\nonce{\verifier}}

\newcommand{\agentdishonest}{c}
\newcommand{\agenta}{a}
\newcommand{\agentb}{b}

\newcommand{\agentc}{c}
\newcommand{\verifier}{V}
\newcommand{\prover}{P}

\newcommand{\protocol}{\mathcal{P}}

\newcommand{\seckey}[1]{sk\ifthenelse{\equal{#1}{\empty}}{}{(#1)}}
\newcommand{\pubkey}[1]{pk\ifthenelse{\equal{#1}{\empty}}{}{(#1)}}
\newcommand{\shakey}[2]{k\ifthenelse{\equal{#1}{\empty}}{}{(#1,#2)}}
\newcommand{\enc}[2]{\left\lbrace #1 \right\rbrace _ {#2}}

\newcommand{\pair}[2]{(#1,#2)}
\newcommand{\concat}{\cdot}

\newcommand{\content}[1]{cont\ifthenelse{\equal{#1}{\empty}}{}{(#1)}}
\newcommand{\actor}[1]{actor\ifthenelse{\equal{#1}{\empty}}{}{(#1)}}
\newcommand{\actors}[1]{actors\ifthenelse{\equal{#1}{\empty}}{}{(#1)}}
\newcommand{\size}[1]{size\ifthenelse{\equal{#1}{\empty}}{}{(#1)}}

\newcommand{\auxmsgm}{l}

\newcommand{\msgm}{m}
\newcommand{\challenge}{C}

\newcommand{\variables}{\mathcal{V}}
\newcommand{\constants}{\mathcal{C}}
\newcommand{\term}{\mathcal{T}_{\signature}(\variables, \constants)}

\newcommand{\makerule}[3]{\frac{\begin{array}{@{}c@{}}#1\end{array}}
			{\begin{array}{@{}c@{}}#2\end{array}}~\rulelabelfont{#3}}

\newcommand{\rulelabelfont}[1]{\texttt{#1}}

\newcommand{\predfont}[1]{\mathit{#1}}

\newcommand{\erasure}{\predfont{erasure}}

\newcommand{\secprop}{\psi}

\newcommand{\realset}{\mathbb{R}}

\newcommand{\startrule}{\ensuremath{\mathsf{Start}}}

\newcommand{\networkdistrule}{\ensuremath{\mathsf{Net}}}
\newcommand{\networkrule}{\ensuremath{\mathsf{Net}}}
\newcommand{\adversaryrule}{\ensuremath{\mathsf{Adv}}}

\newcommand{\infer}{\vdash}

\newcommand{\related}[3]{#1\concat #2 \in #3}

\newcommand{\session}[1]{\Gamma\ifthenelse{\equal{#1}{\empty}}{}{(#1)}}

\newcommand{\myforall}[2]{\forall_{{#1}}{\;#2}}
\newcommand{\myexists}[2]{\exists_{{#1}}{\;#2}}

\newcommand{\z}[1]{\mathsf{#1}}
\newcommand{\eventset}{\z{Ev}}
\newcommand{\send}[2]{\mathrm{send}_{#1}(#2)}  
\newcommand{\recv}[2]{\mathrm{recv}_{#1}(#2)}

\newcommand{\claim}[2]{\mathrm{claim}_{#1}(#2)}

\newcommand{\paramtracesof}[2]{[\hspace{-0.05cm}[#1]\hspace{-0.05cm}]_{#2}}

\newcommand{\tracesof}[1]{[\hspace{-0.05cm}[#1]\hspace{-0.05cm}]}

\newcommand{\advdist}{\delta}

\newcommand{\dbattacker}{distant attacker}
\newcommand{\Dbattacker}{Distant attacker}

\newcommand{\tuple}{(\insym, \outsym, \states, \initialstate, \transfunct, 
\outfunct)}

\newcommand{\statefunction}{\hat{\transfunct}}
\newcommand{\outputfunction}{\hat{\outfunct}}

\newcommand{\insym}{\Sigma}
\newcommand{\outsym}{\Gamma}

\newcommand{\states}{Q}
\newcommand{\initialstate}{q_0}
\newcommand{\transfunct}{\delta}
\newcommand{\outfunct}{\ell}

\newcommand{\universe}{\ensuremath{\mathbf{U}_{\insym,\outsym}}}

\newcommand{\AIn}[2]{\ensuremath{\Omega_{#1}(#2)}}

\newcommand{\pseudofunc}{f}

\newcommand{\maxfunc}{\ensuremath{\operatorname{max}}}
\newcommand{\distfunc}{\ensuremath{\operatorname{d}}}

\newcommand{\speedcom}{\ensuremath{\mathsf{c}}}

\newcommand{\projection}{\pi}

\newcommand{\depth}{d}

\newcommand{\rounds}{n}

\newcommand{\repetitions}{x}

\newcommand{\treeuniverse}{\ensuremath{\mathbf{T}_{\depth}}}

\newcommand{\treeauto}{\ensuremath{{T}^{\outfunct}_{\depth}}}

\newcommand{\protocolIni}{\ensuremath{\protocol_{ini}}}

\newcommand{\fraudprover}{\ensuremath{c}}
\newcommand{\tamarin}{\textsc{Tamarin}}

\newcommand{\distSepPred}[3]{\ensuremath{adv\_sep_{\advdist}(#1,#2, 
#3)}}

\usepackage{amsfonts}
\usepackage{amsmath,amsthm}
\usepackage{amssymb}
\usepackage{graphicx}
\usepackage{url}
\usepackage{blindtext,titlefoot}

\newtheorem{remark}{Remark}
\newtheorem{example}{Example}
\newtheorem{definition}{Definition}
\newtheorem{lemma}[remark]{Lemma}
\newtheorem{theorem}[remark]{Theorem}

\title{Secure Memory Erasure in the Presence of Man-in-the-Middle Attackers
\\
\vspace{0.5cm}
\small (A preprint)}

\author{Rolando Trujillo-Rasua\\ 
{\small School of Information Technology, Deakin University}\\ 
{\small 221 Burwood Hwy., Burwood VIC 3125, Australia}\\ 
{\small rolando.trujillo\@@deakin.edu.au}\\
} 

\date{}

\begin{document}
\maketitle

\begin{abstract}
Memory erasure protocols serve to clean up a device's memory before the 
installation of new software. Although this task can be accomplished by direct 
hardware manipulation, remote software-based memory erasure protocols have 
emerged as a more efficient and cost-effective alternative. Existing remote 
memory erasure protocols, however, still rely on non-standard adversarial 
models to operate correctly, thereby requiring additional hardware to restrict 
the adversary's capabilities. In this work, we provide a formal definition 
of secure memory erasure within a symbolic security model that utilizes the 
standard Dolev-Yao adversary. Our main result consists of a restriction 
on the Dolev-Yao adversary that we prove necessary and sufficient to solve the 
problem of finding a protocol that satisfies secure memory erasure. We also 
provide a description of the resulting protocol using standard cryptographic 
notation, which we use to analyze the security and communication complexity 
trade-off commonly present in this type of protocols. 

\end{abstract}

\section{Introduction}\label{sec:introduction}

Malicious code is a well-known threat to computational devices that support 
a programmable memory. The threat can be mitigated by pro-active 
mechanisms that detect and prevent the installation of malware, viruses, or 
other sort of malicious code. Independently of the success of such defenses, a 
number of devices cannot afford the implementation 
of anti-malware software due to computational and 
operational constraints, e.g. Internet of Things (IoT) devices. Hence, low-cost 
pervasive devices 
rarely come with 
build-in pro-active defenses against malicious code. 

\emph{Memory attestation} is a digital forensics technique used 
to detect whether a device has been compromised by verifying the integrity of  
its memory~\cite{SPDK2004,SLPDK2006}. Devices unable to successfully pass the 
memory 
attestation procedure are regarded as corrupt, and are immediately 
isolated from other devices. A less ambitious, yet often equally 
effective, technique is known as \emph{secure memory erasure}, which 
eliminates  
malicious code that 
resides in memory by fully erasing the device's memory. This is considered an 
important 
preliminary step prior the download and installation of legitimate software. 

Memory erasure in 
itself can be a functional requirement in IoT applications where the ownership 
of devices may change. Erasing the device's memory helps the previous owner to 
protect confidentiality of the information stored in the device, while it gives 
the current owner a proof of absence of malicious software. The latter feature 
is exploited by Perito and Tsudik~\cite{PT2010}, who argue that memory 
erasure is 
a form of memory attestation; both can guarantee the absence of malware.

The memory erasure 
problem can be easily accomplished by accessing the hardware 
directly, but such access is cost ineffective and not scalable~\cite{KK2014}. 
This has opened the door to a number of communication protocols aiming at 
guaranteeing that a given device has actually erased its memory, without 
resorting on hardware manipulation. Memory erasure protocols, also known as 
Proofs of Secure Erasure (PoSE) 
protocols~\cite{PT2010}, allow a remote 
verifier to be convinced that a prover has purged all of its memory. This is 
typically achieved by first depleting the prover's memory with random data, and 
after asking the prover to compute a function on the received data 
as a proof of erasure~\cite{KK2014,ADCH2018}.

We observe that both memory attestation and memory erasure protocols have 
been historically 
designed to operate under non-standard adversarial models, such as the model by 
Francillon et al.~\cite{FNRT2014}, where compromised devices do not reveal 
their secrets, or the models used in~\cite{PT2010,DKW2011,GW2015}, where 
verifier and prover communicate in isolation. Implementing those adversarial 
models is not cost-free, though, as they typically rely on especial hardware to 
make cryptographic keys inaccessible to attackers~\cite{FNRT2014} or 
jamming techniques that selectively block malicious 
communication~\cite{PT2010}.

Recent work~\cite{ADCH2018} starts to advocate for memory 
erasure protocols capable of functioning in the presence of man-in-the-middle 
attackers, arguing that 
selective jamming is ineffective~\cite{PL2012}. We address 
such problem in this article, and make the following 
contributions.

\begin{itemize}
	\item We introduce the notion of \emph{\dbattacker}: a Dolev-Yao 
	adversary~\cite{dolevyao83} restricted to a given distance threshold on 
	their communication with 
	honest protocol participants. 
	\item We restrict the security protocol semantics introduced by Basin et 
	al.~\cite{BCSS2009} and provide a formal definition for secure 
	memory 
	erasure within a symbolic security model. We prove that such restriction is 
	necessary. 
%, by showing that no 
%	secure memory erasure protocol exists in the original semantics.
	\item We provide a high level specification of a memory erasure protocol 
	and prove it secure in the presence of \dbattacker{s}. 
To the best of our 
	knowledge, our protocol is the first memory erasure protocol that can 
	operate in an environment controlled by a Dolev-Yao attacker.  
%	Our proofs rely on the 
%	security protocol verification tool Tamarin~\cite{tamarin13}. 
	\item Lastly, we perform a probabilistic analysis of the security and 
	communication complexity trade-off of the 
	proposed protocol via an 
	instantiation of the high level specification to a 
	class of protocols known as \emph{lookup-based distance-bounding 
	protocols}~\cite{MTT2016}. Protocols of this type have been neglected in 
	literature as they suffer from a security-memory trade-off, i.e. 
	security is proportional to memory demand. However, we obtain the rather 
	surprising result that such a 
	drawback becomes a positive feature in memory erasure protocols.
	
\end{itemize}

The remainder of this article is organized as follows. 
Section~\ref{sec-background} briefly covers the literature on memory 
erasure and memory attestation protocols. Section~\ref{sec-distance-attacker} 
provides an informal introduction to the adversarial model and the definition 
of 
secure memory erasure used throughout the paper. 
Section~\ref{sec-formal-model} and Section~\ref{sec-memory-erasure-def} 
formalize those intuitions within a symbolic security model. 
Section~\ref{sec-lookup-protocols} presents a
high-level specification of a memory erasure protocol that resists 
man-in-the-middle attacks. Finally, Section~\ref{sec-tree-protocol} is devoted 
to the analysis of the security and communication complexity trade-off commonly 
found in this type of protocols. 

\section{Background}
\label{sec-background}
There exists two categories of memory attestation and memory erasure 
techniques~\cite{FNRT2014}. The first one 
relies on special and trusted hardware installed on a device. This technique 
has been regarded as expensive and unsuitable for low-cost IoT 
devices~\cite{SMKK2005}. 
The 
other one is software-based, where few assumptions on the hardware 
capabilities of devices are made.

In a memory attestation procedure, a verifier is capable of reading part or the 
entire prover's memory. If the reading is correct, the verifier can determine 
whether malicious code resides in the prover and take security measures 
accordingly. Obtaining a proof of correct reading is challenging, though. The 
device can be already infected, making it easy for 
malicious code to delude simple reading requests. 

Various memory attestation techniques have been introduced in recent years. 
SWATT~\cite{SPDK2004}, for example, challenges the prover to traverse and 
compute a checksum of its 
memory based on a random total order of memory indexes. For that, the 
verifier uses a 
pseudo-random sequence with a secret seed revealed to the prover prior commence 
of the attestation procedure. SWATT assumes that, with high probability, a 
number of indexes will hit memory occupied by the malicious code. This forces 
the malicious code to make on-the-fly checksum computations, which is detected 
by measuring the delay in the prover's responses. A similar technique is used 
by Pioneer~\cite{SLSPDK2005}.

Shaneck et al. argue that tight timing-based techniques can hardly be used for 
remote 
attestation~\cite{SMKK2005}. The reason is that the network and communication 
channel may introduce unforeseen delays, and that accurately estimating 
computational time in 
software execution is a challenge in itself. Hence, Shaneck et al. propose a 
scheme where the verifier 
sends an obfuscated and self-modifying attestation code to the prover. The 
security of their scheme is based on the observation that performing static 
analysis on a \emph{freshly} generated code is notoriously difficult for an 
attacker.

Secure memory erasure is less ambitious than memory 
attestation in terms of reading capabilities. Yet it can be equally effective 
when it comes to ensuring that a device contains no malicious code in memory. 
%Furthermore, memory erasure finds application in multi-app IoT platforms, 
%where 
%the IoT device's life-cycle goes through multiple 
%stakeholders~\cite{ADCH2018}. 
%
A common assumption in the literature on memory erasure is 
that the prover does not receive external help during the execution of the 
protocol. This is enforced in~\cite{PT2010} by selective jamming of all, but 
the prover, nearby devices during the run of the protocol. Communication 
jamming has the side-effect of preventing man-in-the-middle attackers from 
interfering with the execution of the protocol, hence making security analysis 
simpler.

Improvements upon the protocol in~\cite{PT2010} have mainly focused on 
computational complexity, 
e.g.~\cite{DKW2011,GW2015}, keeping selective jamming as a key 
security measure. However, communication jamming has the drawback of not been 
fully effective~\cite{PL2012}, i.e. it can be bypassed. Moreover, it 
may be 
illegal in some standard wireless bands. It follows the question of whether 
secure memory erasure protocols that resists 
man-in-the-middle attacks can be implemented. 

To the best of our knowledge, SPEED~\cite{ADCH2018} is the first memory erasure 
protocol that aims to 
resist 
man-in-the-middle attacks. It implements a memory 
isolation technique, as in~\cite{Wahbe:1993:ESF:173668.168635},  
to store security-relevant functions and cryptographic 
keys. In this trusted portion of the device's memory, SPEED implements a 
distance bounding protocol~\cite{BC1993} with the goal of enforcing proximity 
between  
prover and verifier. The authors argue that proximity makes it difficult for 
an attacker  
to tamper with the protocol's messages. However, 
the security of their protocol still rely 
on 
assumptions that are not considered in the distance bounding literature, such 
as the use of visual inspection to prevent 
impersonation. 

No memory erasure protocol in literature has been proven secure 
within standard symbolic security models, 
such as the Dolev-Yao model~\cite{dolevyao83}. In contrast, security standards 
(e.g. {ISO/IEC} 
9798~\cite{BasinCM13}) and major security protocols (e.g. 
TLS 1.3~\cite{CremersHHSM17}) have been analyzed, fixed and improved, by 
expressing their goals within symbolic security models and verifying their 
correctness 
with automated proving tools, such as
ProVerif~\cite{proverif} and Tamarin~\cite{tamarin13}. 
%and 
%Scyther~\cite{Cremers08}. 
This work addresses such gap. 

\section{Secure memory erasure in the presence of \dbattacker 
s}\label{sec-distance-attacker}

In this section, we introduce an informal security framework to analyze  
memory erasure protocols in the presence of man-in-the-middle attackers. A 
formalization within a symbolic security model of the concepts introduced in 
this section will follow immediately after.

\subsection{Secure memory erasure}

Most proofs of secure erasure (PoSE) are based on 
the notion of memory filling, whereby a verifier requests a prover to full 
its memory with meaningless data, such as a random sequence. In this setting,  
Karvelas and 
Kiayias~\cite{KK2014} consider a memory erasure protocol secure if 
the prover cannot maintain a portion of its memory intact. 
Perito and Tsudik's definition is more fine-grained, stating that 
secure erasure is achieved when prover and verifier agree on a memory 
variable~\cite{PT2010}. We adopt in this article the latter. 

\begin{definition}[Secure memory erasure]\label{def-sec-erasure}
Let $V$ and $P$ be a verifier and prover, respectively. Let $V_{mem}$ be a 
variable stored in $V$'s memory, and $P_{mem}$ a variable stored in $P$'s 
memory. A protocol satisfies \emph{secure memory erasure} if, for every 
successful 
execution of the protocol, there exists an execution step where $V_{mem} = 
P_{mem}$ and the size of $V_{mem}$ is 
equal to $P$'s writable memory. 
\end{definition}

%Note that, in the definition above, agreement on a variable value is not 
%sufficient to guarantee secure erasure. In addition, the value stored in the 
%variable ought to be sufficiently large. 

In the absence of a man-in-the-middle 
attacker, most memory erasure protocols satisfy 
Definition~\ref{def-sec-erasure}. As discussed by Perito and 
Tsudik's~\cite{PT2010}, even a simple protocol where the verifier sends a 
random nonce 
and expects to receive 
the same nonce back as a proof of erasure 
satisfies Definition~\ref{def-sec-erasure} for a sufficiently large nonce. 
Therefore, it remains to introduce the adversarial model against which the 
security of 
this type of protocols can be assessed.

%\begin{figure}\centering
%\includegraphics[scale=0.65]{../images/perito-tzudik-protocol}
%\caption{Perito and Tsudik's protocol~\cite{PT2010} with a single nonce 
%exchange.}
%\label{fig-perito-prot}
%%\lessspace
%\end{figure}

%\begin{figure}\centering
%\begin{tabular}{c c }
%\includegraphics[scale=0.8]{images/perito-tzudik-protocol} &\qquad
%\includegraphics[scale=0.8]{images/perito-tzudik-protocol-attack} \\
%(a) &\qquad (b) 
%\end{tabular}
%\caption{(a) Perito and Tsudik's protocol~\cite{PT2010} with a single nonce 
%exchange and (b) a man-in-the-middle attack.}
%\label{fig-perito-prot}
%%\lessspace
%\end{figure}

\subsection{The adversarial model}

In security models, adversaries are characterized in terms of 
their ability to compromise agents and manipulate network 
traffic. 
While various notions of agent compromise 
exist~\cite{CanettiK2001,BasinC2014}, 
%such as long-term 
%key 
%compromise, session key compromise and state retrieval, 
the adversary's 
capabilities to interact with the network are, with few exceptions, still those 
introduced 
by Dolev and Yao in 1982~\cite{dolevyao83}. That is, an adversary with the 
ability to block, modify, and inject arbitrary messages in the network.

In memory erasure protocols, the prover may have malicious 
code running in memory. This allows a Dolev-Yao attacker, also known as 
man-in-the-middle attacker, to easily 
impersonate the prover, making any intention of interaction with the 
prover meaningless. That is the reason why memory erasure protocols have been 
traditionally designed to operate within a clean environment, where no 
attacker is able to manipulate the network. Such a clean environment has been 
traditionally enforced by radio jamming~\cite{PT2010}. As illustrated in 
Figure~\ref{fig-jamming}, a man-in-the-middle attacker 
can 
be frustrated by allowing the prover to complete the protocol while 
selectively jamming the attacker's signal. Even if the adversary is within the 
jamming radius, he can neither interact with the prover nor with the verifier. 

\begin{figure}\centering
\includegraphics[scale=.2]{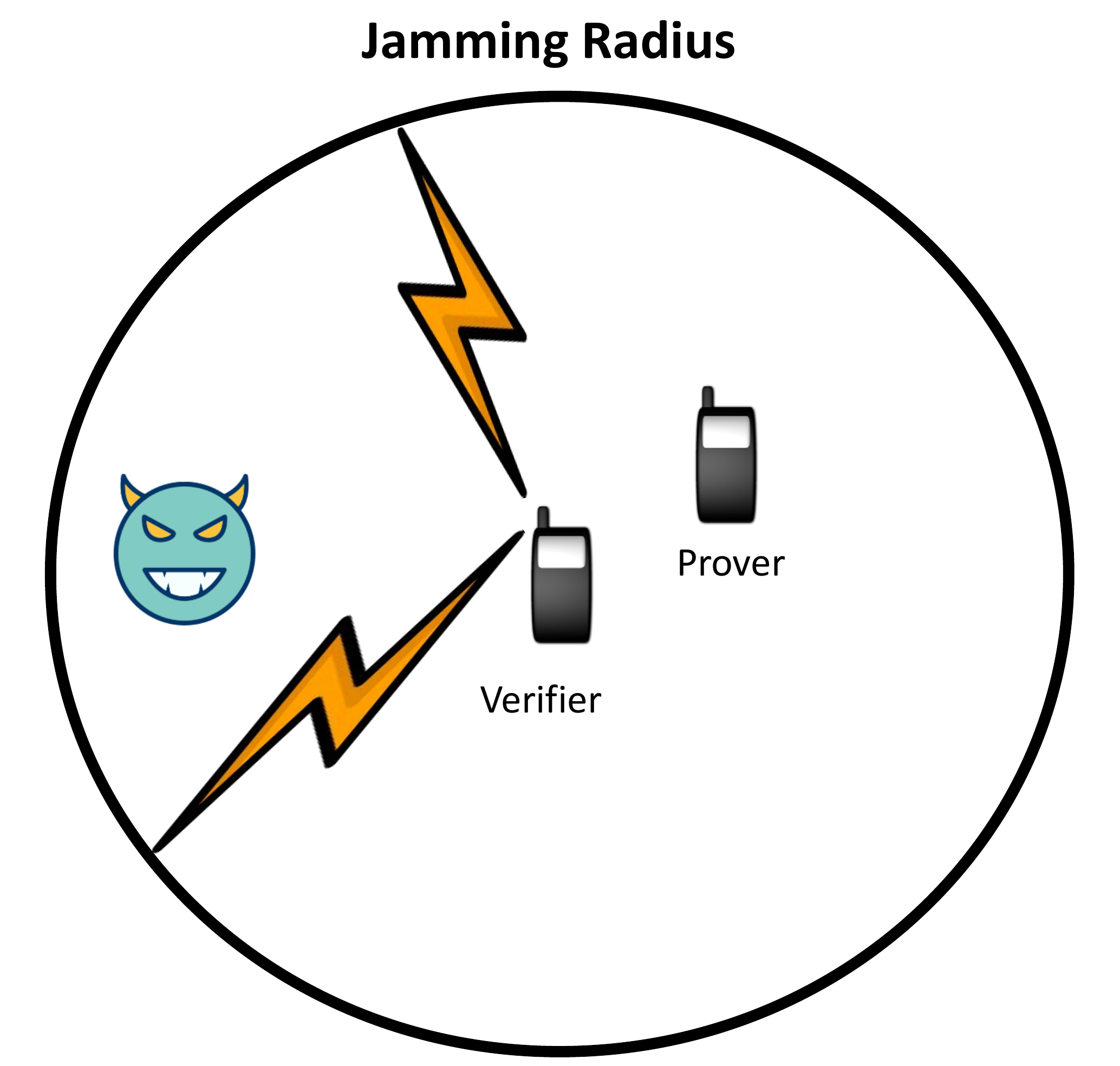} 
\caption{Preventing man-in-the-middle attacks via jamming.}
\label{fig-jamming}
%\lessspace
\end{figure}

%\begin{figure*}\centering
%\begin{tabular}{c c}
%\includegraphics[scale=.3]{images/jamming} &\qquad
%\includegraphics[scale=.3]{images/distance-bounding}\\
%(a) & \qquad (b) 
%\end{tabular}
%\caption{Two mechanisms to frustrate man-in-the-middle.}
%\label{fig-boundary}
%%\lessspace
%\end{figure*}

The idea of creating an area where a man-in-the-middle attacker looses his 
capabilities was taken a step further by Ammar et al.~\cite{ADCH2018}. They 
proposed the use of a distance bounding protocol~\cite{BC1993} to ensure 
proximity between prover and verifier. Intuitively, if the 
interaction between prover and verifier is limited to a sufficiently 
small area, as in Figure~\ref{fig-boundary}, the attacker is thwarted from 
participating in the protocol execution. In a sense, Ammar et al. suggest that 
distance bounding protocols can be used as a primitive to weaken 
man-in-the-middle adversaries and simplify the design and analysis of security 
protocols. We make this notion more precise next.

\begin{figure}\centering
\includegraphics[scale=.2]{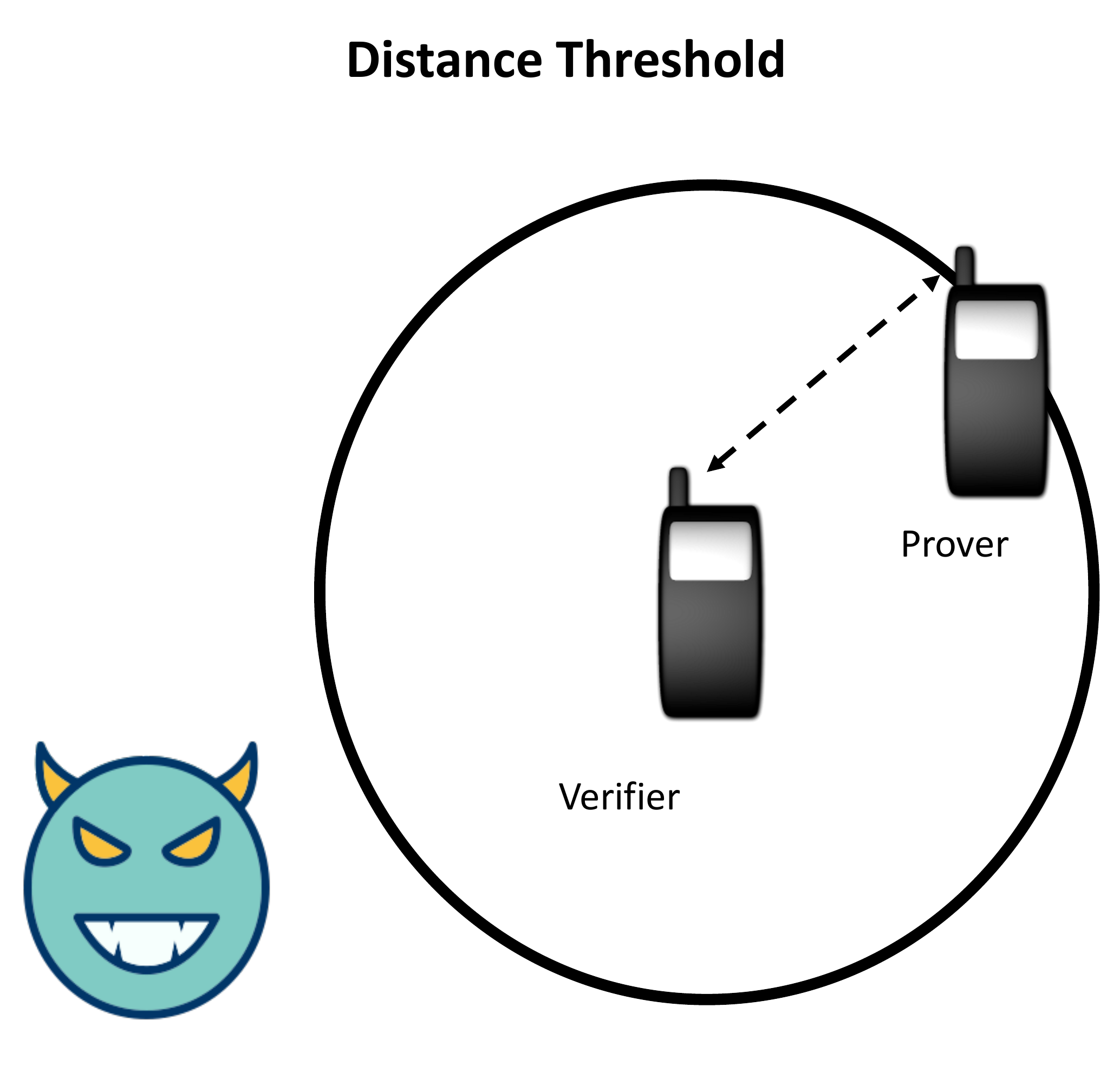}
\caption{Preventing attacks from distant attackers. }
\label{fig-boundary}
%\lessspace
\end{figure}

\begin{definition}[\Dbattacker]\label{def-dist-adv}
Given a distance threshold $\advdist$, a \emph{\dbattacker} is a 
Dolev-Yao attacker whose distance to the verifier is larger 
than or equal to $\advdist$.
\end{definition}

Clearly, the larger $\advdist$ the weaker a \dbattacker{} is with respect to 
the standard Dolev-Yao adversary. Nonetheless, we point out that assuming a 
\dbattacker{} is reasonable in various applications where the 
protocol is executed in a controlled environment, such as a private room. 
The challenge is to design a memory erasure protocol that resists 
attacks from a \dbattacker{} with $\advdist > 0$. 

Our goal next is to formalize the intuitions exhibited in this section and 
prove the following two propositions. First, under standard assumptions in 
symbolic 
security protocol verifications, no protocol can resist a distant attacker with 
$\advdist = 0$. 
%Second, 
%the 
%SPEED protocol~\cite{ADCH2018}, which is the first memory erasure protocol 
%that aims to correctly operate in the presence of a man-in-the-middle 
%attacker, is insecure against the distant attacker defined above and also 
%against restricted notions of that attacker. 
Second, for every $\advdist > 0$, 
there exists a security 
protocol that can be proven secure in the presence of a distant attacker 
with distance threshold $\advdist$. 

%We show 
%next that SPEED does not achieve the intended security goal, though. 

\section{The Security Model}\label{sec-formal-model}

To formalize the notion of \dbattacker, we need a model that supports reasoning 
about distance between protocol participants. A security model of this type has 
been  
introduced by Basin et al.~\cite{BCSS2009}, provided with a 
trace-based semantics for distributed systems describing all possible 
interleaved events that protocol principles can execute.

%\subsubsection{The Syntax}\label{sec:model-syntax}
%\subsection{Protocol specification}

\subsection{Messages, events and inference rules}

\noindent \emph{Messages.}
%\headingfourthlevel{Messages} 
A security protocol defines the way various 
protocol participants, called \emph{agents}, exchange cryptographic messages. 
To model cryptographic messages, we use a term algebra $\term$ where 
$\signature$ is a signature,  $\variables$ a set of 
variables, and $\constants$ a set of constants. 
%We reserve in $\constants$ the 
%standard symbols used for real numbers 
%to represent distance values, such as $0, 1, 2, ...$, and we use them with the 
%standard interpretation. 
We consider agents' names, 
denoted $\agentset$, and nonces, denoted $\nonceset$, to be constants in our 
term algebra as well, i.e. $\agentset, \nonceset \subseteq \constants$. 
The set of nonces is assumed to be partitioned 
into $\{\nonceset_{\agenta}\mid \agenta\in\agentset\}$. This is to ensure that 
two different agents cannot produce the same nonce. The set $\agentset$ itself 
is also 
partitioned into $\honestset$ (honest agents) and $\dishonestset$ (dishonest 
agents). 
Finally, we assume 
that the 
signature $\signature$ contains the following function symbols:

\begin{itemize}
\item $pair(\msgm, \msgm')$ denoting the pairing of two 
terms $\msgm$ and $\msgm'$. We will usually use $\pair{\msgm}{\msgm'}$ as 
shorthand notation. 

\item $enc(m, k)$ denoting the encryption of $\msgm$ 
with the key $k$. We will usually use $\enc{\msgm}{k}$ as shorthand notation.
%\item $\seckey{\agenta}$ and $\pubkey{\agenta}$ denote the long-term secret and
%public key of the agent 
%$\agenta\in\agentset$, respectively.
\item $\shakey{\agenta}{\agentb}$ denoting the long-term shared symmetric 
secret key of two agents 
$\agenta,\agentb\in\agentset$.
%\item $\invkey{key}$ denotes the decryption key of the key $key\in\msgset$.
\end{itemize}

We use $\msgset$ to denote  the set of all terms obtained from the above term 
algebra.

It is worth remarking that we have intentionally omitted asymmetric encryption 
in our term algebra. The reason is that the protocols we analyze in this 
article make no use of public keys. That said, our results can be easily 
extended to a model that supports public-key encryption.

\noindent \emph{Events and traces.}
%\headingfourthlevel{Events and traces} 
Agents can execute three types of 
events: i) send a message, ii) receive a message and iii) claim that a 
given security property holds. Hence we define the set 
of all possible events $\eventset$ by the following grammar.         
\begin{align*}
e \grammarequal \send{\agenta}{\msgm}\mid 
\recv{\agenta}{\msgm}\mid 
\claim{\agenta}{\secprop,\msgm}\text{,}
\end{align*}
where $\agenta$ is an agent's name, $\msgm$ a message, 
and $\secprop$ a constant identifying a 
security property. We consider the auxiliary function $\actor{}\colon 
\eventset\to \agentset$, which 
provides the actor executing an event. 

\begin{align*}
& \actor{e} = \agenta \iff \\
& \hspace{1cm} e \equiv \send{\agenta}{\msgm} \vee 
e \equiv \recv{\agenta}{\msgm} \vee e \equiv 
\claim{\agenta}{\secprop,\msgm} \text{,}
\end{align*}

%\todo{Consider whether this auxiliary functions can be introduced later}

When constructing traces, each event is given a time-stamp $t \in \realset$, 
representing the time at which the event has been executed. Therefore, 
a trace is a finite sequence of time-stamped events  
$\trace = (t_1, e_1) \cdots (t_n, e_n ) \in (\realset \times \eventset)^*$. In 
this case, we say that $\trace$ has cardinality $n$, denoted $|\trace|$, and we 
use $\trace_i$ to denote the $i$th element of $\trace$, i.e. $\trace_i = (t_i, 
e_i)$. The auxiliary function $\maxfunc(\trace)$ gives the largest time-stamp 
in $\trace$, while  the function $\actor{}$ is extended 
to time-stamped events in the straightforward way.

\noindent \emph{Inference.} 
%\headingfourthlevel{Inference} 
We formalize the way agents obtain and create 
knowledge by an inference 
relation $\infer\;\subseteq \agentset\times (\realset \times \eventset)^*\times 
\msgset$. And, we use the shorthand notation $\agenta \infer_{\trace} \msgm$ to 
denote $(\agenta, \trace, \msgm) \in \infer$,  indicating that agent $a$
can infer message $m$ from trace
$\trace$. We define the relation 
$\infer$ to be the least set that is closed under the inference rules in 
Figure~\ref{fig:inference-rules}. Each of these rules states that:

\begin{itemize}
	\item \emph{Rule $\rulelabelfont{I1}$}: except other agent's nonces, an 
	agent $\agenta$ can infer any constant, including its own set of nonces 
	$\nonceset_a$. 
	\item \emph{Rule $\rulelabelfont{I2}$}: agents can infer 
	their shared secret keys with other agents. 
	\item \emph{Rule $\rulelabelfont{I3}$}: all function symbols in 
	$\signature$, but $\shakey{}{}$, can be used to infer arbitrary 
	terms constructed over already inferable terms. The 
function symbol $\shakey{}{}$ is reserved to be used only in rule 
$\rulelabelfont{I2}$.
	\item \emph{Rule $\rulelabelfont{I4}$}: a receive event 
	$\recv{\agenta}{\msgm}$ allows agent $a$ to infer $m$.
	\item \emph{Rule $\rulelabelfont{I5}$}: agents have the ability to unpair 
	messages. 
	\item \emph{Rule $\rulelabelfont{I6}$}: an encrypted message $\enc{x}{y}$ 
	can be decrypted with the decryption key $y$.
\end{itemize}

%It is worth remarking that the above inference rules are standard for 
%models formalizing idealized encryption.

\begin{figure}

\[
\arraycolsep=1.7pt\def\arraystretch{1.7}
\begin{array}{c c}
\makerule{\msgm\in
(\constset \setminus \nonceset) \cup \nonceset_\agenta 
}{\agenta\infer_{\trace}\msgm}{I1} \\
\makerule{}{\agenta\infer_{\trace}
(\shakey{\agenta}{\agentb}, \shakey{\agentb}{\agenta})}{I2}\\ 
\makerule{\agenta\infer_{\trace}\msgm_1, \ldots, \agenta\infer_{\trace}\msgm_n, 
f\in\signature\setminus\privfunset}
{\agenta\infer_{\trace}f(\msgm_1, \ldots, \msgm_n)}{I3} \\
\makerule{(t, \recv{\agenta}{\msgm}) 
\in\trace}{\agenta\infer_{\trace}\msgm}{I4} 
\hspace{1cm}
\makerule{\agenta\infer_{\trace}\pair{x}{y}}{\agenta\infer_{\trace}x, 
\agenta\infer_{\trace}y}{I5} \\
\makerule{\agenta \infer_{\trace} 
\enc{x}{y}, \agenta \infer_{\trace} 
y}{\agenta\infer_{\trace} x}{I6}  \\
\end{array}
\]
%\[
%\arraycolsep=1.7pt\def\arraystretch{1.7}
%\begin{array}{c c}
%\makerule{\msgm\in
%(\constset \setminus \nonceset) \cup \nonceset_\agenta 
%}{\agenta\infer_{\trace}\msgm}{I1} &
%\makerule{}{\agenta\infer_{\trace}
%(\shakey{\agenta}{\agentb}, \shakey{\agentb}{\agenta})}{I2}\\ 
%\makerule{\agenta\infer_{\trace}\msgm_1, \ldots, 
%\agenta\infer_{\trace}\msgm_n, 
%f\in\signature\setminus\privfunset}
%{\agenta\infer_{\trace}f(\msgm_1, \ldots, \msgm_n)}{I3} & 
%\makerule{(t, \recv{\agenta}{\msgm}) 
%\in\trace}{\agenta\infer_{\trace}\msgm}{I4} 
%\\
%\makerule{\agenta\infer_{\trace}\pair{x}{y}}{\agenta\infer_{\trace}x, 
%\agenta\infer_{\trace}y}{I5} & 
%\makerule{\agenta \infer_{\trace} 
%\enc{x}{y}, \agenta \infer_{\trace} 
%y}{\agenta\infer_{\trace} x}{I6}  \\
%\end{array}
%\]
\caption{Inference rules}
\label{fig:inference-rules}
\end{figure}

\subsection{A security protocol semantics}

\noindent \emph{Protocol specification.} 
%\headingfourthlevel{Protocol specification} 
The following protocol specification assumes that agents remain in a static 
location during the protocol execution. This is a standard assumption in 
security models dealing with physical properties, such 
as~\cite{BCSS2009,basinCSF09,MTT2016}. An uninterpreted distance 
function 
$\distfunc(.)$ is used to denote the distance between agents $\agenta$ 
and $\agentb$.

A protocol $\protocol$ is defined 
by a set of derivation rules, similar to the inference rules used above,  
specifying how execution traces make progress. Its semantics, denoted 
$\tracesof{\protocol}$, is the least set closed under 
those rules. We now describe the type of rules 
we use to inductively define the set of traces of a protocol. The base case is 
modeled by the start rule ($\startrule$), which 
indicates that the empty trace $\emptytrace$ is in $\tracesof{\protocol}$ for 
every protocol $\protocol$. 

\[
\makerule{}{\emptytrace \in \tracesof{\protocol}}{\startrule}\text{,}
\]

The other two rules, $\networkdistrule$ and $\adversaryrule$, are used to model 
the network behavior and corruption capability of the adversary. 
The $\networkdistrule$ rule states that a message $\msgm$ 
previously sent by $\agenta$
can be received by any agent $\agentb$ whose distance to $\agenta$ is 
consistent with the physical relation between constant speed, distance, and 
time. That is to say, given the propagation speed $\speedcom$ of the 
communication channel, it 
must hold that a message sent by $\agenta$ at time $t$ and received by 
$\agentb$ at time $t'$ satisfies $d(a, b) \leq 
\frac{\speedcom}{2}(t' - t)$. 

\begin{align*}
& \makerule{\trace \in \tracesof{\protocol}, (t, \send{\agenta}{\msgm}) \in 
\trace \\
t' \geq \maxfunc(\trace), \distfunc(\agenta, \agentb) \leq 
\frac{\speedcom}{2}(t' - t), 
}
{\related{\trace}{(t', 
\recv{\agentb}{\msgm})}{\tracesof{\protocol}}}{\networkdistrule}\text{,}
\\
\end{align*}

Lastly, the $\adversaryrule$ rule allows the adversary to 
impersonate dishonest agents 
and send events on their behalf. Note that, unless otherwise specified, 
variables in derivation rules are universally quantified.

\begin{align*}
& \makerule{\trace \in \tracesof{\protocol}, 
\agenta \in \dishonestset, \\ 
t \geq \maxfunc(\trace), \agenta \infer_{\trace} \msgm }
{\related{\trace}{(t, 
\send{\agenta}{\msgm})}{\tracesof{\protocol}}}{\adversaryrule} 
\text{,}
\end{align*}

The rules $\startrule$, $\networkdistrule$ and $\adversaryrule$, are part of 
every protocol specification, which are used 
to model 
trace initialization, network operation, and 
the behavior of dishonest agents. It remains to define the behavior of honest 
agents. Such behavior is defined by protocol-dependent rules with 
the following constraints. 

\begin{itemize}
\item References to dishonest agents in either the premises or the conclusion 
of a protocol rule are not allowed. The reason is that the behavior of 
dishonest agents is fully specified by the $\adversaryrule$ rule.
\item All events in a rule ought to be executed by the same agent. That is, the 
logic of an agent's behavior is independent from the specification of other 
agents. Hence agents interact exclusively through the $\networkrule$ rule.
\item Terms used in the conclusion of a protocol rule ought to be derivable 
from the premises by using the inference rules in 
Figure~\ref{fig:inference-rules}. 
\end{itemize}

\begin{example}
We use a simplified version of the memory erasure protocol introduced by 
Perito and 
Tsudik~\cite{PT2010}. In 
the protocol, the verifier sends a challenge $n$ and the prover reflects it 
back to the verifier\footnote{This is a simplification of the original 
protocol~\cite{PT2010} used for illustration purposes only. }. 
%The goal here is not 
%analyzing the security of their 
%protocol, but illustrating how it can be modeled within the present formalism 
%with three derivation rules. 
Its specification is given in 
Figure~\ref{fig-perito-prot-spec}, where rule $\ruleName{V1}$ 
states that $V$ can choose any of its own nonces and send it through. For the 
sake of simplicity, we are assuming in protocol rules that nonces are 
fresh. Hence, when we 
use the premise $N_V \in \nonceset_V$ we are also requiring that $N_V$ has not 
been used in the trace $\trace$. 
Rule 
$\ruleName{P1}$ indicates that, upon reception of a nonce $N_V$, $P$ sends 
$N_V$ back. Finally, $\ruleName{V2}$ is used by $V$ to claim 
that a given \emph{erasure} property should hold after reception of the nonce 
$N_V$. Further below we explain claim events in more detail.
\end{example}

\begin{figure}\centering
\includegraphics[scale=0.8]{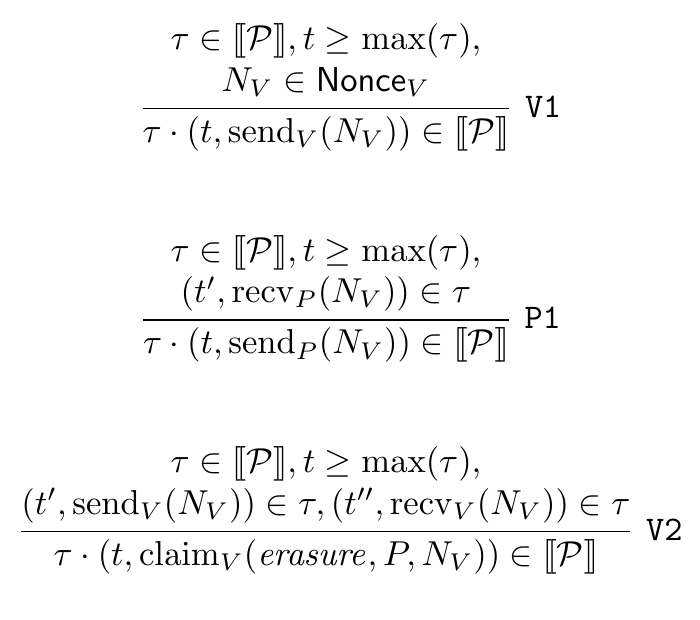} 
\caption{Specification of Perito and 
Tsudik's protocol. }
\label{fig-perito-prot-spec}
%\lessspace
\end{figure}

\noindent \emph{Execution traces and security properties.} 
%\headingfourthlevel{Protocol execution} 
An execution 
trace of 
a protocol $\protocol$ is any trace inductively defined by the set of rules 
$\{\startrule, \networkdistrule, \adversaryrule\} \cup \protocol$. For example, 
the 
protocol specification in Figure~\ref{fig-perito-prot-spec} gives the following 
trace via application of the $\startrule$, $\ruleName{V1}$, 
$\networkrule$, $\ruleName{P1}$, $\networkdistrule$ and $\ruleName{V2}$ rules, 
where $\agenta, \agentb \in \agentset$ and $n \in \nonceset$.

%In this trace we assume that both $\agenta$ and $\agentb$ are honest, so the 
%value of $\advdist$ has no effect on the execution. 

\begin{align*}
& \emptytrace \concat (0, \send{\agenta}{\anonce}) \concat (1, 
\recv{\agentb}{\anonce})
\concat 
(2, \send{\agentb}{\anonce}) \concat  \\
& (3, \recv{\agenta}{\anonce}) \concat (4, 
\claim{\agenta}{\erasure, 
\agentb, 
\anonce}) \text{,}
\end{align*}

Finally, a security property is a first-order logic statement on execution 
traces, which 
is said to be satisfied by a protocol $\protocol$ if the property holds for all 
traces in $\tracesof{\protocol}$. 
To account for the fact that a 
security property 
needs not be satisfied over 
the entire execution of a protocol, we are using claim events as
placeholders to indicate those execution steps where a security property needs
to be verified. This is, for example, the role of the claim 
event  $\claim{\verifier}{\erasure, \prover, N_V}$ in the protocol 
specification in 
Figure~\ref{fig-perito-prot-spec}. In this case, the verifier $\verifier$ 
expects that, upon reception of the nonce $N_V$, the prover $\prover$ has 
received the nonce $N_V$.

\section{An underapproximation of secure memory 
erasure}\label{sec-memory-erasure-def}

%Once the formal security model has been established, we are ready to 
%formalize secure memory erasure by 
%restricting the 
%semantics of Basin et 
%al.'s model. We show 
%the necessity of the new semantics by proving that no protocol satisfies 
%secure memory erasure in the original semantics. 

%\subsection{
%Formalizing secure memory erasure}

We consider traces where a verifier interacts with a prover to obtain a proof 
of secure erasure, with the restriction that any external help the prover can 
receive comes from attackers that are located at a distance from the verifier 
no 
lower than a given threshold $\advdist$. The following predicate is used to 
determine whether a trace $\trace$ satisfies such condition, with respect to a 
verifier 
$\agenta$ and a prover $\agentb$.

\begin{align*}
\distSepPred{\trace}{\agenta}{\agentb} \iff 
& \myforall{\agentc \in \actors{\trace}} \agentc \neq \agentb \implies \\
\quad \distfunc(\agenta, \agentc)
\geq \advdist \vee \agentc \in \honestset \text{,}
\end{align*}

where  $\actors{\trace} = \cup_{(t, 
e) \in \trace} \{\actor{(t, e)}\}$.

%\begin{definition}\label{def-prot-sem}
%The semantics of $\protocol$ in the presence of a 
%\dbattacker{} with distance threshold $\advdist$, denoted 
%$\disttracesof{\protocol}$, is defined by,
%
%\begin{align*}
%& \disttracesof{\protocol} = \{\trace \in \tracesof{\protocol} | \myforall{ 
%\agenta, \agentb \in 
%\actors{\trace}}{} \agenta \in \honestset \wedge \\
%& \hspace{1cm} \agentb \in 
%\dishonestset \implies \distfunc(\agenta, \agentb) 
%\geq 
%\advdist\} \text{,}
%\end{align*}
%
%where $\actors{\trace} = \cup_{(t, e) \in \trace} \{\actor{(t, e)}\}$. 
%\end{definition}

Secure memory erasure is defined below as a statement on traces satisfying the 
attacker separation property, rather than on the full protocol semantics. 
Intuitively, if at some step of an execution trace $\trace$ an agent $\agenta$ 
believes that another agent $\agentb$ has erased 
its memory by storing a (sufficiently large) message $\msgm$, then it must be 
the case that 
$\agentb$ previously received or sent a message $\msgm$.
%that was not part of $\agentb$'s knowledge (memory) at the 
%start of the protocol. 

%We use this new restricted semantics to 
%provide a formal definition of secure memory erasure based on claim events of 
%the type  $\claim{\verifier}{\erasure, \prover, M}$, where  
%$P$, $V$ and $M$ are  
%variables in the term algebra $\term$. Our definition uses the notion 
%that messages (terms in our algebra) may be composed of \emph{smaller} 
%messages. This is formally captured by the syntactic subterm relation 
%$\subterm$ over 
%$\term$ defined as the reflexive, transitive closure of the smallest relation 
%satisfying that, for every $\msgm_1, \ldots, \msgm_n$ and every 
%$f\in\signature\setminus\privfunset$:
%
%\[
%\arraycolsep=5pt\def\arraystretch{1.7}
%\begin{array}{c c}
%\msgm_1 \subterm (\msgm_1, \msgm_2), & \msgm_2 \subterm (\msgm_1, \msgm_2), \\
%\msgm_1 \subterm \enc{\msgm_1}{\msgm_2}, & \msgm_2 \subterm 
%\enc{\msgm_1}{\msgm_2}, \\ 
%\multicolumn{2}{c}{\myforall{{i \in \{1, \ldots, n\}}} \msgm_i \subterm 
%f(\msgm_1, \ldots, \msgm_n)} \text{.} \\ 
%\end{array}
%\]

\begin{definition}[Secure memory erasure]\label{def-formal-sec-erasure}
Let $\protocol$ be a protocol. The claim event $\claim{\verifier}{\erasure, 
\prover, M}$  is said to be correct in $\protocol$ with respect to 
a \dbattacker{} with 
distance threshold $\advdist$ if, 
\begin{align*}
%& \myexists{\trace \in 
%\disttracesof{\protocol}, \agenta, \agentb \in \agentset, \msgm \in \msgset}{} 
%\claim{\agenta}{\erasure,\agentb,\msgm} \in 
%\trace \wedge \\
& \myforall{\trace \in 
\tracesof{\protocol}}{} (t, 
\claim{\agenta}{\erasure,\agentb,\msgm}) 
\in 
\trace \wedge \agenta \in \honestset \wedge \\
& \hspace{1cm} \distSepPred{\trace}{\agenta}{\agentb}
\implies \myexists{t' < t}{} (t', \recv{\agentb}{\msgm}) \in \trace \vee \\
& \hspace{2cm} 
(t', 
\send{\agentb}{\msgm}) \in 
\trace 
%\agentb \infer_{\trace} \msgm \wedge \agentb \not \infer_{\emptytrace} \msgm 
 \text{,}
\end{align*}
\end{definition}

%\RT{Definition 1 says that a protocol satisfies secure memory erasure if
%  at some point, the memory of V and P coincides, intuitively ensuring
%  that the memory of P is eventually filled with some random value
%  chosen by V. I do not see how this is reflected in Definition 4
%  which says that some term m is known for ever  by some agent
%  b. Since Definition 4 is the key definition of the paper, a thorough
%  discussion, back up with examples of attacks should be provided.}

Observe that the action of receiving or 
sending a message $\msgm$ 
is considered a guarantee that an agent has or had $\msgm$ in memory. This 
indeed resembles the standard notion of agreement ~\cite{lowe97hierarchy} in 
security protocols. Moreover, the prover is allowed to be dishonest. 
This is 
less common in security properties expressed within a symbolic model, but a key 
assumption in the memory erasure scenario. 

It is worth remarking that 
Definition~\ref{def-formal-sec-erasure} 
underapproximates the 
intuition given in 
Definition~\ref{def-sec-erasure}, as it neglects 
the size of the term $\msgm$. We make the assumption that the size of the term 
prover and verifier agree upon is known by the analyst, e.g. $\msgm$ is a 
1024-bit 
nonce, and that such size is large enough to deplete the prover's memory. This 
means that optimizations on the size of the memory required to store $\msgm$ 
ought to be 
analyzed out of the introduced model, as we do further below in 
Section~\ref{sec-tree-protocol}.

%And third,
%requiring that the message 
%$\msgm$ is not inferable by the prover at the beginning of the protocol is 
%used 
%to rule out memory erasure protocols that keeps the prover's memory immutable. 

%This choice is justified by the fact that 
%reasoning about message size is still challenging in symbolic security models. 
%Although Definition~\ref{def-formal-sec-erasure} is sufficient to 
%detect 
%logical flaws in variants of published protocols (as we show next), we 
%acknowledge that further research is needed to come up with a 
%tighter approximation of Definition~\ref{def-sec-erasure} that can still allow 
%for formal reasoning. 

\subsection{Analyzing a variant of SPEED}

To illustrate how the proposed definition can be used, 
we analyze a variant of SPEED~\cite{ADCH2018}. This choice is based on the fact 
that SPEED is, to the best of our knowledge, the first memory erasure protocol 
that measures the distance between prover and verifier, which is a property 
that can be 
exploited to prevent man-in-the-middle attacks. We remark, nonetheless, that 
SPEED was thought to resist a definition of security weaker than that in 
Definition~\ref{def-formal-sec-erasure}. Our analysis below serves for 
illustration purposes only and does not diminish the merits of the protocol. 

The SPEED protocol, depicted in Figure~\ref{fig-SPEED-prot}, starts when the 
verifier $V$ sends the hash of a random 
bit-sequence $m_1 \cdots m_n$ to the prover $P$. Upon reception of the hash 
value, $P$ executes $n$ rounds of rapid bit exchanges, known as \emph{the fast 
phase}~\cite{BC1993}, where the prover measures the 
round-trip-time of several bit exchange rounds with the 
verifier. At the $i$th round of the fast phase, 
$P$ chooses a random bit $a_i$ 
and sends it to $V$. Then $V$ immediately responds with $r_i = a_i \oplus m_i$. 
The round-trip-time $\Delta t_i$ of the bit exchange is 
calculated by $P$ upon receiving $r_i$, allowing $P$ to verify that $\Delta 
t_i$ is below a given threshold. $P$ 
also computes 
the 
bit-sequence $m_1' \cdots m_n'$ where $m_i' = r_i \oplus a_i \ \forall i \in 
\{1, \ldots, n\}$, and checks that $h(m_1' \cdots m_n') = h(m_1 \cdots m_n)$. 
If both verification steps are passed successfully, $P$ erases its memory with 
a default value $MeM$ and sends a MAC computation on the protocol transcript 
and the internal memory of the prover. 

\begin{figure}\centering
\includegraphics[scale=0.70]{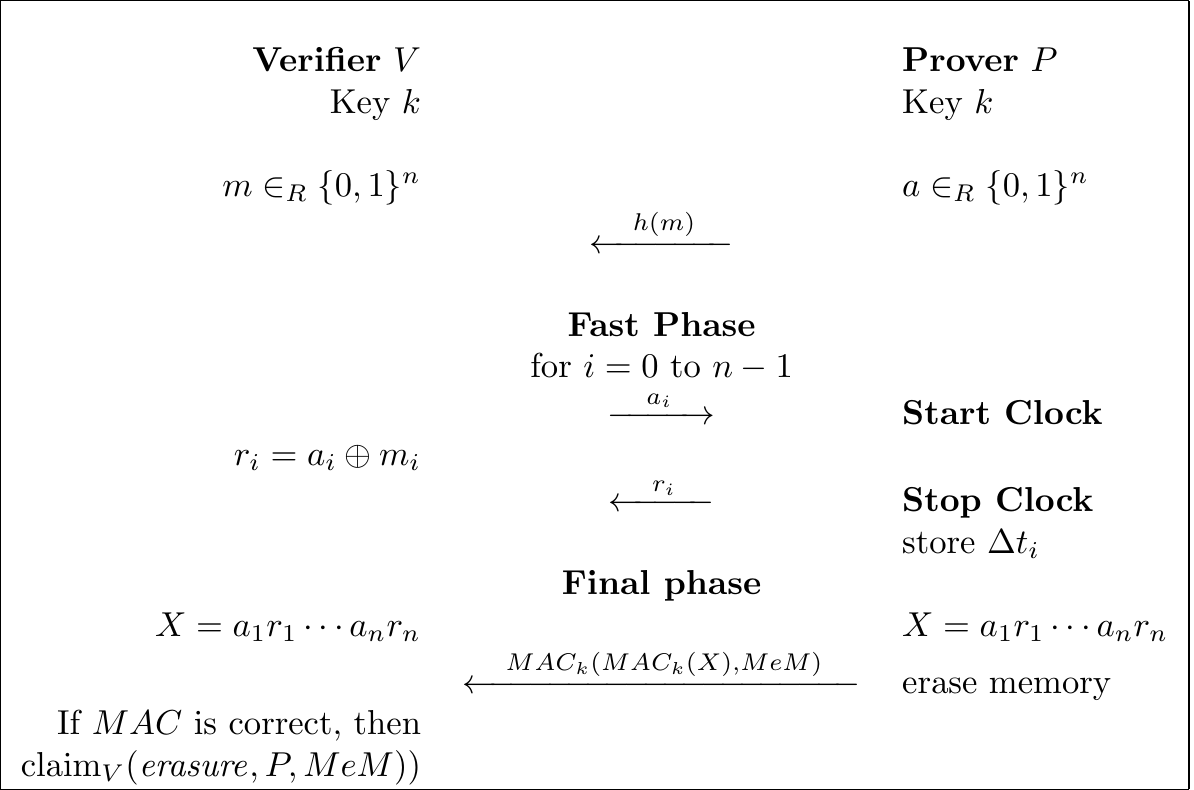}
\caption{The SPEED protocol with shared keys.}
\label{fig-SPEED-prot}
%\lessspace
\end{figure}

We note that, because the original design of SPEED does not use cryptographic 
keys, it ought to rely on offline 
methods, such as visual inspection, to fight against impersonation attacks. 
Given that neither visual inspection nor any other type of offline method is 
considered by the security model introduced herein, we strength the protocol by 
assuming a MAC function that uses a shared secret key between prover and 
verifier. 

%\RT{P4, Section III.C "breaking the SPEED protocol". It seems to me that
%  the protocol was never intended to allow a verifier to evaluate the
%  distance of the prover. So I am not convinced that the proposed
%  behaviour is an attack. The authors should clarify the initial
%  security expectations of the SPEED protocol.}

The security analysis of SPEED given 
in~\cite{ADCH2018} is based on the following four main 
assumptions. First, the 
prover does not 
execute sessions in parallel with verifiers. Second, the adversary cannot 
tamper 
with the security-relevant functions of the prover. Third, the cryptographic 
primitives and sources of randomness used in the protocol are secure. And 
fourth, the 
round-trip-time calculations can be used by the prover to enforce proximity 
with the verifier~\cite{BC1993}. 
%That is, considering $c$ to be the maximum 
%network transmission speed and $\Delta t$ a 
%round-trip-time computation, SPEED expects the prover-to-verifier distance to 
%be lower than or equal to 
%$\frac{1}{2}\Delta t \cdot c$.  
We, 
nevertheless, can construct an attack trace that satisfies those assumptions 
and invalidates Definition~\ref{def-formal-sec-erasure}. 

The attack trace (depicted in Figure~\ref{fig-SPEED-attack}) consists of an 
adversary impersonating a prover, with the peculiarity that the prover is 
willingly contributing to the attack by revealing its cryptographic keys, i.e. 
the prover is dishonest. It starts 
when the verifier $V$ aims at erasing the memory of a prover $P$, for which  
a random sequence $m \in_{R} \{0, 1\}^n$ of size $n$ is generated. $V$ sends 
the hash of $m$ to $P$, which should be used later by $P$ to check 
proximity with 
$V$. At this point, an adversary $A$ takes over the communication with $V$ and 
replies to $h(m)$ by executing $n$ rounds of the fast phase, as established by 
the protocol. We assume the adversary-verifier communication occurs at an 
arbitrary distance $\delta$, and that $P$ is voluntarily not taking part in 
the protocol execution. The adversary chooses to challenge the verifier with a 
sequence 
of 
zeros. At the end of the 
fast phase, 
$A$ replies with $MAC_k(MAC_k(0 || m_1 \cdots 0 || m_n), MeM)$ where $MeM$ is 
the 
default value the prover is expected to use to erase its memory. Such MAC value 
is correct, making $V$ incorrectly believe that $P$ has erased its memory by 
storing the 
value $MeM$.

%In addition to the attack trace above, we provide a formal specification of 
%SPEED in the protocol verification 
%tool \tamarin{}~\cite{tamarin13} at 
%\url{https://github.com/memory-erasure-tamarin/code}\footnote{As usual in 
%symbolic verification, the implemented protocol is an abstraction of SPEED. }, 
%which the 
%reader can use 
%to verify SPEED in an automated way. 

The presented attack is based on the simple observation that a dishonest prover 
can provide the adversary (i.e. another dishonest device) with the 
cryptographic material necessary to pass the memory erasure protocol. This 
means that a single rogue device can be used to bypass the memory erasure 
procedure of many compromised devices, that is, the attack scales easily. 
Moreover, this type of 
external help is not ruled out by the security model, because 
the adversary complies with the restriction of being far enough from the 
verifier. 

\begin{figure}\centering
\includegraphics[scale=0.70]{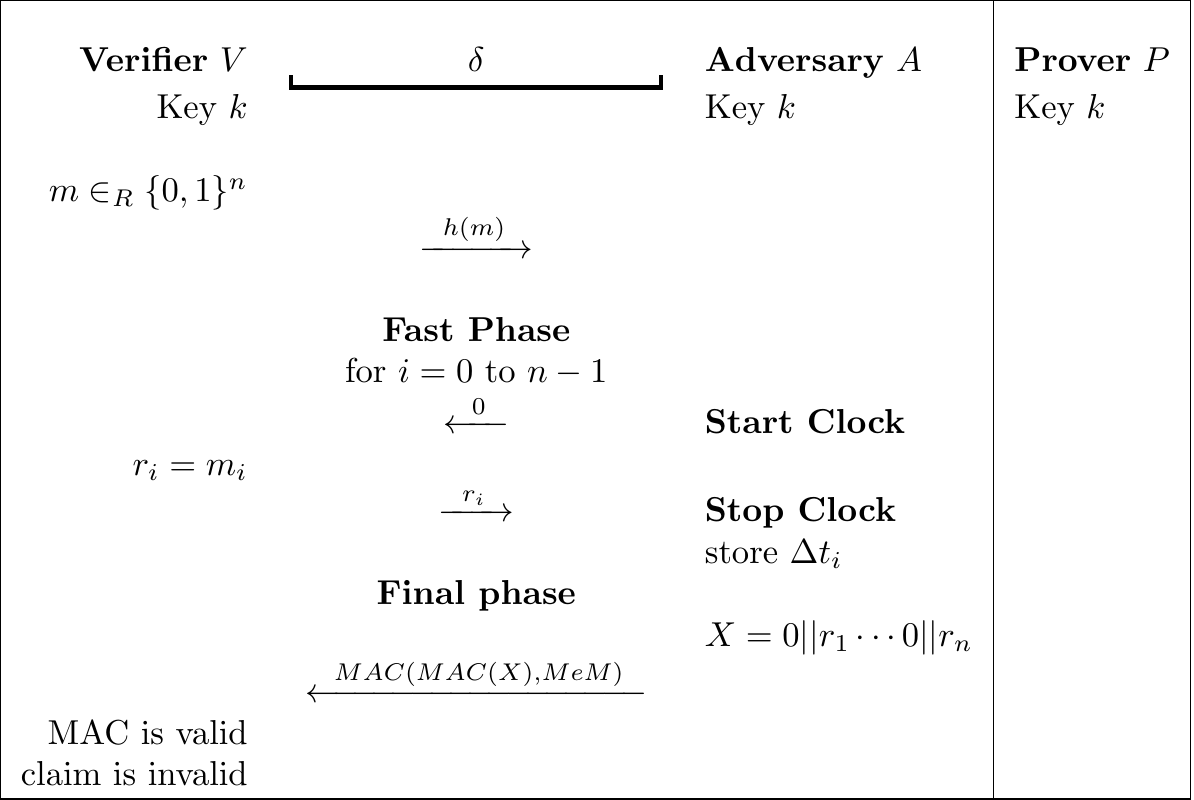}
\caption{Attack on SPEED with shared keys.}
\label{fig-SPEED-attack}
%\lessspace
\end{figure}

%Given that SPEED does not meet the intended security goal, we regard the 
%question of whether a memory erasure protocol can be proven secure against a 
%\dbattacker{} still open. 
%We dedicate the remainder of this article to 
%introduce a class of memory erasure protocols than can be proven secure within 
%a formal security model.

%Three remarks are worth mentioning. 
%First,  Third, we assume that the claim event 
%$\claim{\verifier}{\erasure, \prover, M}$  is reachable. 
%Such reachability 
%property is typically considered a sanity check in a protocol specification 
%and 
%often omitted from security property definitions. 

\subsection{Impossibility result towards secure memory erasure}

Before presenting an alternative to SPEED, we deem important providing an 
impossibility result on the problem of finding a 
protocol that satisfies secure memory erasure when $\advdist = 0$, as this 
proves the necessity of the restriction on the distance between the adversary 
and the verifier.

\begin{theorem}
Let $\advdist = 0$. Then for every protocol $\protocol$ and trace $\trace 
\in \paramtracesof{\protocol}{}$, 
\begin{align*}
& (t, \claim{\agenta}{\erasure,\agentb,\msgm}) \in 
\trace \implies \\
& \hspace{0.5cm} \exists_{\trace' \in \tracesof{\protocol}} 
(t, \claim{\agenta}{\erasure,\agentb,\msgm}) 
\in 
\trace' 
%\wedge 
%\agentb \not 
%\infer_{\emptytrace} \msgm  
\wedge \distSepPred{\trace'}{\agenta}{\agentb} \wedge \\
& \hspace{1cm} \left ( \forall_{t' < t} (t', \recv{\agentb}{\msgm'}) \not \in 
\trace' \wedge (t', 
\send{\agentb}{\msgm'}) \not \in 
\trace' \right ) 
\text{,}
\end{align*}

\end{theorem}

\begin{proof}
%\RT{the proof of Theorem 1 essentially amounts to the existence of a 
%man-in-the-middle attacker with access to the prover's secrets, yet this is 
%never actually mentioned in plain terms. The proof is fairly trivial, but 
%presented this way, it hides the main insight.}

%The proof is based on the observation that 
%Definition~\ref{def-formal-sec-erasure} does not require $\agentb$ to be 
%honest.
Consider a trace $\trace \in \paramtracesof{\protocol}{}$ such that it contains 
a 
claim 
event $(t, \claim{\agenta}{\erasure,\agentb,\msgm})$. We observe that $\trace$ 
can be constructed based on a partition of the set $\agentset = \honestset \cup 
\dishonestset$ with $\agentb 
\in \dishonestset$. Hence we consider another dishonest agent 
$\agentdishonest \in \dishonestset$ such 
that 
$\distfunc(\agentb, \agentdishonest) = 0$ and $\agentdishonest  \not \in 
\actors{\trace}$. 
That is, 
both $\agentb$ and $\agentdishonest$ are in the same location and 
$\agentdishonest $ has not 
been active in $\trace$. 
Then we 
construct the trace $\trace'$ as follows, for every $i \in \{1, \ldots, 
|\trace|\}$ and every $\auxmsgm \in \msgset$, 
\[
  \trace'_i = \left\{
  \begin{array}{@{}ll@{}}
    (t', \send{\agentdishonest}{\auxmsgm}), & \text{if}\ \trace_i = (t', 
    \send{\agentb}{\auxmsgm}) \\
    (t', \recv{\agentdishonest}{\auxmsgm}), & \text{if}\ \trace_i = (t', 
    \recv{\agentb}{\auxmsgm}) \\
	\trace_i, & \text{otherwise.} \\
  \end{array}\right.
\]

Now, let $t_{0}$ be the minimum time-stamp of an event in $\trace$. We create 
the following trace, 

\begin{align*}
& \trace'' = (t_0, \send{\agentb}{(\shakey{\agenta}{\agentb}, 
\shakey{\agentb}{\agenta})})) \concat \\ 
& \hspace{0.5cm}
(t_0, \recv{\agentdishonest}{(\shakey{\agenta}{\agentb}, 
\shakey{\agentb}{\agenta})}))  \concat 
\trace'
\end{align*}

The trace $\trace''$ consists of $\agentb$ revealing its secret key with 
$\agenta$, followed by $\agentdishonest$ learning the keys 
$\shakey{\agenta}{\agentb}$ and $\shakey{\agentb}{\agenta}$. The remaining 
events 
in $\trace''$ are those in $\trace'$ respecting the original order. We will 
prove that $\trace'' \in \tracesof{\protocol}$. 
For that, we use 
$\trace[i]$ to denote the sub-trace $\trace_1 \concat \cdots \concat \trace_i$ 
with $i \in \{1, \ldots, |\trace|\}$ and $\trace[0]$ to denote the empty trace 
$\emptytrace$. Then we prove via 
induction that for every $i \in \{0, \ldots, |\trace|\}$ and every $\auxmsgm 
\in 
\msgset$, 

\begin{align}\label{form-inference}
& \agentb \infer_{\trace[i]} \auxmsgm \wedge \auxmsgm \not \in 
\nonceset_{\agentb} 
\implies 
\agentdishonest 
\infer_{\trace''[i+2]} 
\auxmsgm \\
& \trace[i] \in \tracesof{\protocol} \implies \trace''[i+2] \in 
\tracesof{\protocol} \text{,} \label{form-traces}
\end{align}

\noindent \emph{Base case $[i = 0]$. } Notice that if $\agentb 
\infer_{\trace[0]} \auxmsgm \wedge \auxmsgm \not \in \nonceset_{\agentb}$, then 
either 
$\auxmsgm$ is a constant or 
$\auxmsgm = (\shakey{\agenta}{\agentb}, 
\shakey{\agentb}{\agenta})$, according to the derivation rules in 
Figure~\ref{fig:inference-rules}. The second event in $\trace''$ 
allows $\agentdishonest$ to infer $(\shakey{\agenta}{\agentb}, 
\shakey{\agentb}{\agenta})$ (Rule \rulelabelfont{I4}). And, if $\auxmsgm$ is a 
constant and not a nonce, then all agents can infer $\auxmsgm$ (Rule 
\rulelabelfont{I1}). 
Therefore, $\agentb \infer_{\trace[0]} \auxmsgm\wedge \auxmsgm \not \in 
\nonceset_{\agentb} 
\implies \agentdishonest \infer_{\trace''[2]} \auxmsgm$.
Now, given that $\distfunc(\agentb, \agentdishonest) = 0$, it 
follows that 
$(t_0, \send{\agentb}{(\shakey{\agenta}{\agentb}, 
\shakey{\agentb}{\agenta})})) \concat 
(t_0, \recv{\agentdishonest}{(\shakey{\agenta}{\agentb}, 
\shakey{\agentb}{\agenta})})) $ is in $\tracesof{\protocol}$. 

\noindent \emph{Inductive step.} 
We now assume that (\ref{form-inference}) and (\ref{form-traces}) hold for 
every $i \in \{0, \ldots, k\}$ with $k < |\trace|$. We analyze two cases. 

\noindent \emph{Case 1: $\trace_{k+1} = (t', \recv{\agentb}{\auxmsgm'})$ for 
some 
time-stamp 
	$t'$ 
	and term 
$\auxmsgm'$.} In this case, $\trace_{k+3}'' = (t', 
\recv{\agentdishonest}{\auxmsgm'})$ by construction of the trace $\trace''$. As 
a 
result, it holds that $\agentb \infer_{\trace[k+1]} \auxmsgm \wedge 
\agentdishonest 
\not \infer_{\trace''[k+3]} \auxmsgm \implies \agentb \infer_{\trace[k]} 
\auxmsgm 
\wedge \agentdishonest 
\not \infer_{\trace''[k+2]} \auxmsgm$, which proves 
(\ref{form-inference}) for 
$i 
= k+1$ by contrapositive. To 
prove that (\ref{form-traces}) holds for $i = k+1$, we just need 
to notice that $\trace_{k+1}$ is appended to $\trace[k]$ via application of the 
rule $\networkdistrule$. Because all events, but those of $\agentb$, in 
$\trace$ are 
preserved in $\trace''$, and $\agentb$'s events are now mimic by 
$\agentdishonest$, rule $\networkdistrule$ can also be applied to 
append $(t', 
\recv{\agentdishonest}{\auxmsgm'})$ to $\trace''[k+2]$, which gives 
$\trace''[k+3] 
\in \tracesof{\protocol}$.

\noindent \emph{Case 2: $\trace_{k+1} \neq (t', \recv{\agentb}{\auxmsgm'})$ for 
every 
	time-stamp 
	$t'$ 
	and term $\auxmsgm'$.} In this case, it follows that 
$\agentb \infer_{\trace[k]} \auxmsgm \iff \agentb 
\infer_{\trace[k+1]} \auxmsgm$. By hypothesis of induction we thus obtain that 
$\agentb \infer_{\trace[k+1]} \auxmsgm \wedge 
\auxmsgm \not \in 
\nonceset_{\agentb} 
\implies \agentdishonest \infer_{\trace''[k+2]} \auxmsgm \implies 
\agentdishonest 
\infer_{\trace''[k+3]} \auxmsgm$, which proves the induction step for 
(\ref{form-inference}). It remains to prove that 
(\ref{form-traces}) holds for $i = k+1$, for which we analyze two more
cases. 
\begin{itemize}
	\item $\trace_{k+1} = (t', \send{\agentb}{\auxmsgm'})$ for some time-stamp 
	$t'$ 
	and term $\auxmsgm'$. The rule $\adversaryrule$ is used to append 
	$\trace_{k+1}$ to $\trace[k]$, meaning that $\agentb \infer_{\trace[k]} 
	\auxmsgm'$. By (\ref{form-inference}) we obtain that $\agentdishonest 
	\infer_{\trace''[k+2]} \auxmsgm$. Hence via application of the rule 
	$\adversaryrule$ we obtain that  $\trace''[k+3] \in \tracesof{\protocol}$.
	\item The last case is $\trace_{k+1} = \trace''_{k+3}$, which occurs when 
	$\actor{\trace_{k+1}} \neq b$. Because all events in $\trace[k]$, except 
	those 
	from $\agentb$, are preserved in $\trace''[k+2]$, then $\trace[k] 
	\concat \trace_{k+1} \in \tracesof{\protocol} \implies \trace''[k+2] 
	\concat \trace''_{k+3} \in \tracesof{\protocol} $.
\end{itemize}
We conclude that $\trace'' \in 
\tracesof{\protocol}$. Now, by construction, $\trace''$ also satisfies that 
$(t, 
\claim{\agenta}{\erasure,\agentb,\msgm}) \in \trace''$ and that 
no $t'$ and $\auxmsgm$
exist such that $(t', \recv{\agentb}{\auxmsgm}) \in \trace''$ or $(t', 
\send{\agentb}{\auxmsgm}) \in \trace''$. Given that $\advdist = 0$,  
$\distSepPred{\trace''}{\agenta}{\agentb}$ also holds. This yields the expected 
result. 
\end{proof}

The main corollary of the theorem above is that no protocol can satisfy secure 
memory erasure with $\advdist = 0$. This corresponds to the case where no 
restriction to a distant attacker is imposed. We thus conclude that restricting 
traces 
to a given 
separation between honest and dishonest agents is necessary towards the goal of 
finding a secure memory erasure protocol. We prove in the next section that 
such restriction is also sufficient.

\section{A secure memory erasure protocol}
\label{sec-lookup-protocols}
%\todo{to see where to put this text...below}
%
%First, because the prover can be compromised, it can pre-empt to the 
%verifier's 
%challenges and break the proximity assumption. This is known as a 
%distance-fraud attack~\cite{BC1993}. Second, distance bounding protocols do 
%not 
%guarantee the absence of adversary actions within the verifier's proximity, 
%but 
%instead that some messages of the protocol cannot be relayed. This means that 
%the total absence of adversaries ought to be verified by other means, for 
%example, visual inspection. Last but not least, distance bounding protocols 
%work by timing the round-trip-time of message exchanges. Not all messages are 
%timed. Hence the prover can still receive external help when replying 
%to un-timed messages. 
%
%
%Designing a protocol that satisfies secure memory erasure and secure distance 
%bounding is clearly a challenge. The rest of this section is devoted to 
%formally defining what we mean by secure memory erasure and secure distance 
%bounding. 
%\todo{to see where to put this text...above}

In this section we introduce a memory erasure protocol that can be 
proven secure within the security model introduced in earlier sections. 
This is, to the best of our knowledge, the first memory erasure protocol that 
resists man-in-the-middle attacks.

\subsection{The protocol}

%\RT{The authors propose a framework, claiming generality and then
%  analyse their second protocol outside their own model.}

We propose a protocol that aims at mutual authentication between prover and 
verifier. The need of authenticating the verifier is to 
prevent illegitimate erasure requests, while prover authentication is a 
necessary step towards obtaining a proof of erasure. In addition, the verifier 
measures the round-trip-time of a message exchange to obtain a bound on its 
distance to the prover. This is a distance bounding technique~\cite{BC1993} 
that will prove useful to counteract distant attackers. 

The protocol, depicted in Figure~\ref{fig-prot-symbolic}, starts when the 
prover $\prover$ sends a nonce $\nonceP$. A verifier $\verifier$ replies with a 
Message Authentication Code (MAC) on the nonce 
$\nonceP$ and a freshly generated nonce $\nonceV$, which is used by $\prover$
to authenticate $\verifier$. As usual, the MAC function 
is computed using $P$ and $V$'s shared key $k$. Right after, a 
time-measurement phase commences, where  $\verifier$ starts 
a clock and sends 
a challenge $\challenge$ to $\prover$. Immediately after receiving the 
challenge, 
$\prover$ replies with $r = h(k, \nonceP, \nonceV, \challenge)$, where $h$ is a 
hash function. Upon reception of the prover's response, $\verifier$ stops the 
clock 
and calculates the round-trip-time $\Delta_t$. Then $\verifier$ checks that 
$r$ is correct and that $\Delta_t \leq \Delta$, where $\Delta$ is a protocol 
parameter denoting a maximum time-delay tolerance. If both verification steps 
succeed, $\verifier$ claims that $\prover$ has erased its memory.

\begin{figure}\centering
\includegraphics[scale=0.8]{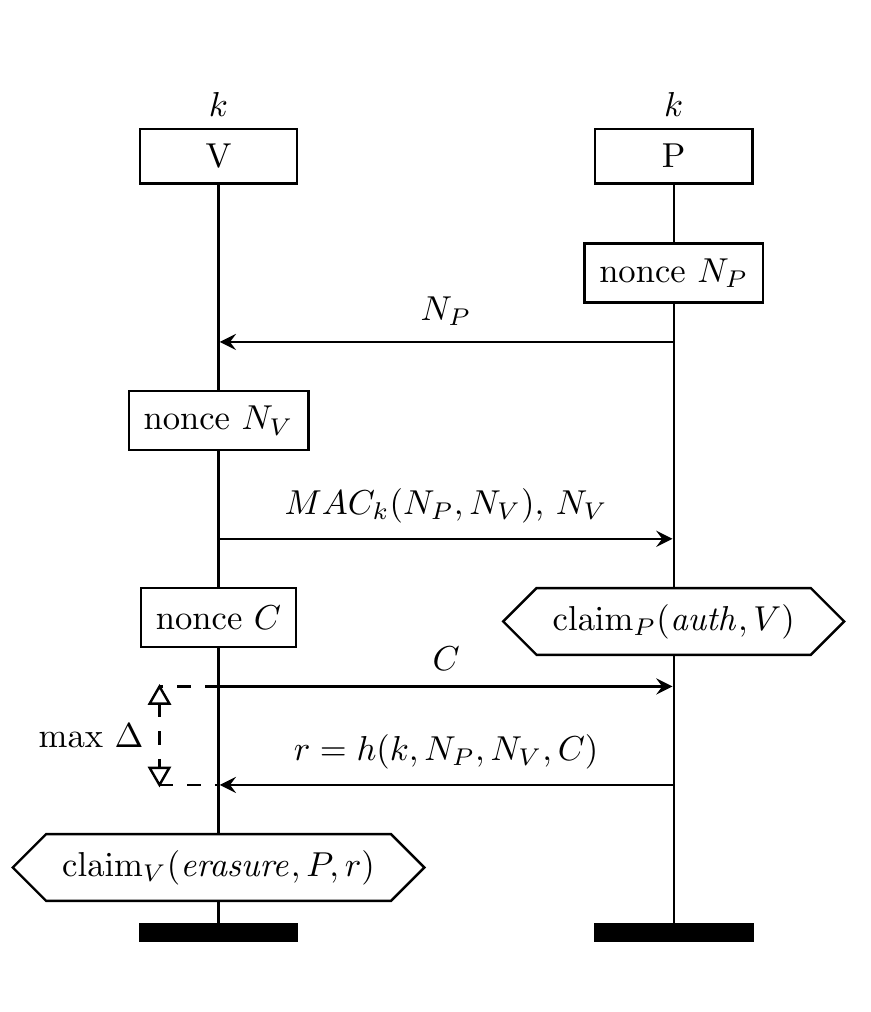}
\caption{A secure memory erasure protocol. }
\label{fig-prot-symbolic}
%\lessspace
\end{figure}

The reader may have noticed that the 
introduced memory erasure protocol does not use standard notation from 
the literature in distance bounding, where the time-measurement phase is 
composed of various rounds of 
single-bit 
exchanges~\cite{BC1993,dbsurvey}. For the moment we require this high level of 
abstraction to come up with formal security proofs. Nonetheless, in 
Section~\ref{sec-tree-protocol} below we unfold the proposed protocol and 
describe it using standard cryptographic notation for distance bounding 
protocols. 

\subsection{Security analysis}

Figure~\ref{fig-prot-spec} provides a formal
specification of the protocol in the modeling language introduced earlier. 
That specification is used to enunciate the various results that come next. In 
the remainder of this section, we use $\protocol$ to refer to the protocol 
defined by the protocol rules in Figure~\ref{fig-prot-spec}.

\begin{figure}\centering
\includegraphics[scale=0.75]{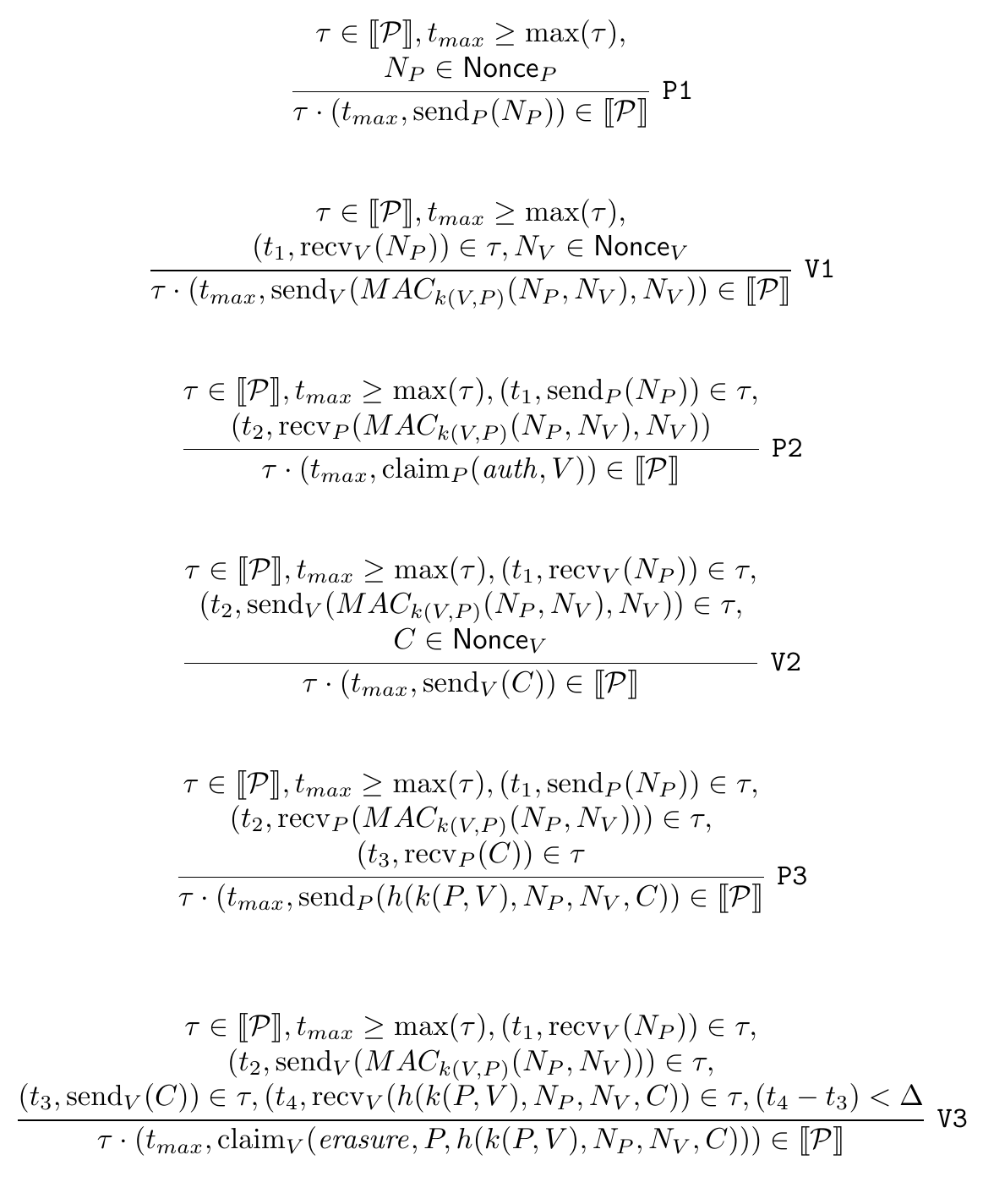} 
\caption{Specification of the introduced memory erasure protocol. }
\label{fig-prot-spec}
%\lessspace
\end{figure}

\begin{lemma}\label{lem-sec-mem-erasure}
%Let $\protocol$ be the protocol defined in Figure~\ref{fig-prot-spec}, and 
%$\traceprojsof{\protocol}$ its semantics. 
Let $\projection: (\realset \times 
\eventset)^{*} \rightarrow \eventset^{*}$ be a projection function on 
time-stamped traces defined by 
$\projection((t_1, e_1) \cdots (t_n, e_n)) = e_1 \cdots e_n$, for every $(t_1, 
e_1) \cdots (t_n, e_n) \in (\realset \times 
\eventset)^{*}$. Let $\projection(\tracesof{\protocol}) = \{\projection(\trace) 
| \trace \in \tracesof{\protocol}\}$. 
$\protocol$
satisfies that

\begin{align}\label{eq-causality}
& \forall \traceproj \in \projection(\tracesof{\protocol}), \agenta, 
\agentb \in 
\agentset, 
n, m, c \in \msgset
\colon \\ \nonumber
& \hspace{0.2cm} r = h(k(\agenta, \agentb), n, m, 
c) \wedge \claim{\agenta}{\erasure, \agentb, r} \in 
\traceproj \wedge \\ \nonumber
& \hspace{0.5cm} \agenta \in \honestset \implies \myexists{i, j, k \in \{1, 
\ldots, |\traceproj|\}, \agentb' \in \agentset} i < k < j 
\wedge 
\\\nonumber
& \hspace{1.0cm} 
\traceproj_i = \send{\agenta}{c} \wedge \traceproj_k = \send{\agentb'}{r} 
\wedge \traceproj_j = \recv{\agenta}{r} \wedge \\\nonumber
& \hspace{1.0cm} \left (\agentb = \agentb'  \vee \{\agentb, 
\agentb'\} \subseteq \dishonestset \right)
\text{,}
\end{align}
\end{lemma}

\begin{proof}
We use the security protocol verification tool 
\tamarin{}~\cite{tamarin13} to prove this lemma. 
The \tamarin{} specification of the protocol 
and lemma can be found at 
\url{https://github.com/memory-erasure-tamarin/code}. Therefore, the 
correctness of 
this proof relies on the claim that the provided \tamarin{} implementation 
faithfully corresponds to the formalization provided herein.
\end{proof}

Lemma~\ref{lem-sec-mem-erasure} states that either the prover (if the prover is 
honest) or a dishonest agent on behalf of the prover (if the prover is 
dishonest) will respond to the challenge sent by the verifier to calculate the 
round-trip-time. 

Our main observation at this point is that the condition satisfied by the 
memory erasure claim event in Lemma~\ref{lem-sec-mem-erasure} is stronger 
than the condition required to satisfy causality-based secure distance 
bounding, as introduced in~\cite{causalDB18}. This allows us to prove the main 
theorem of this section.

\begin{theorem}\label{theo-sec-erasure}
Let 
$\Delta$ 
be the time upper bound used in $\protocol$, $\delta$ a distance threshold 
of a \dbattacker, and $\speedcom$ the transmission speed of the communication 
channel. If $\advdist \geq \frac{\speedcom}{2}\Delta$, then 
$\protocol$ 
satisfies secure memory erasure. 
\end{theorem}

\begin{proof}
According to Definition~\ref{def-formal-sec-erasure}, it is sufficient to prove 
that 
\begin{align*}
%& \myexists{\trace \in 
%\disttracesof{\protocol}, \agenta, \agentb \in \agentset, \msgm \in \msgset}{} 
%\claim{\agenta}{\erasure,\agentb,\msgm} \in 
%\trace \wedge \\
& \myforall{\trace \in 
\tracesof{\protocol}}{} (t, \claim{\agenta}{\erasure,\agentb,r}) 
\in 
\trace \wedge \agenta \in \honestset  \wedge \\
& \hspace{0.5cm} \distSepPred{\trace}{\agenta}{\agentb}\implies
%\agentb \not 
%\infer_{\emptytrace} r  \wedge 
\myexists{t' < t}  (t', \send{\agentb}{r}) \in 
\trace 
\text{,} 
\end{align*}

%Given $\trace = (t_1, e_1) \cdots (t_n, e_n)$ and $(t_i, e_i) = (t, 
%\claim{\agenta}{\erasure,\agentb,r})$ for some $\ell \in \{1, \ldots, n\}$, 
%let 
%$\trace' = (t_1, e_1) \cdots (t_{\ell}, e_{\ell})$. Note that, because the 
%protocol semantics is inductively defined, then $\trace' \in 
%\tracesof{\protocol}$. 
Consider two agents, $\agenta$ and $\agentb$ with 
$\agenta \in \honestset$, 
and a trace 
$\trace \in \tracesof{\protocol}$ such 
that $\distSepPred{\trace}{\agenta}{\agentb}$. 
Consider now the 
projection 
$\traceproj = \projection(\trace)$ 
of the trace $\trace$. It follows that $\traceproj$ satisfies that 
$\claim{\agenta}{\erasure, 
\agentb, r} \in \traceproj$ with $r = h(k(\agenta, 
\agentb), n, m, c)$. This allows 
us to 
use 
Lemma~\ref{lem-sec-mem-erasure} and conclude that there must exist $i, j, 
k \in \{1, \ldots, |\traceproj|\}$ such that $i < k < j$, $\traceproj_i = 
\send{\agenta}{c}$, $\traceproj_j = 
\recv{\agenta}{r}$ and $\traceproj_k = \send{\agentb'}{r}$ for some agent 
$\agentb'$, which is either $\agentb$ itself or dishonest. Mapping back 
those events onto the trace $\trace'$, we obtain that there exists timestamps 
$t_i < t_k < t_j$ such that 
$(t_i, \send{\agenta}{c}) \in \trace$, $(t_k, \send{\agentb'}{r}) \in \trace$
and $(t_j, \recv{\agenta}{r}) \in \trace$. 

By looking at the protocol rules (concretely $\rulelabelfont{V3}$) and the fact 
that 
$\agenta$ is honest, it 
follows that $t_i < t$ and that $t_j - t_i < \Delta$. This means that 
$d(\agenta, 
\agentb') \leq \frac{\speedcom}{2}(t_j - t_i) < \frac{\speedcom}{2}\Delta \leq 
\advdist$, i.e. $d(\agenta, 
\agentb') < \advdist$. On the one hand, because $\trace$ satisfies 
$\distSepPred{\trace}{\agenta}{\agentb}$, we obtain that either $b = b'$ or 
$b'$ is honest. On the other hand, 
Lemma~\ref{lem-sec-mem-erasure} gives that either $b = b'$ or 
$b'$ is dishonest. Given that $b'$ cannot be honest and dishonest at the same 
time, we conclude that $b = b'$. As a result,  $(t_k, \send{\agentb}{r}) \in 
\trace$, which concludes the proof.
\end{proof}

Theorem~\ref{theo-sec-erasure} proves that the proposed memory erasure 
protocol (depicted in Figure~\ref{fig-prot-spec}) resists man-in-the-middle 
attacks from a \dbattacker{} with 
distance threshold 
$\advdist \geq \frac{\speedcom}{2}\Delta$. That is, the protocol 
does not 
contain logical flaws with respect to the mathematical model and properties 
introduced in this article. 

The next and last section of this article is dedicated to analysing attacks 
that are not regarded as man-in-the-middling in the traditional sense; hence 
not captured by the security model. We refer to probabilistic attacks that aim 
to bypass the protocol without fully storing the term $r$ in memory. Of course, 
this requires switching from symbolic analysis to probabilistic analysis.

%\section{Memory requirement of lookup-based memory erasure protocols}
\section{A protocol based on cyclic tree automata}\label{sec-tree-protocol}

%\RT{Section VI restates much of what has already been covered in previous 
%sections, just recasting it with automata instead of abstract hashes.}

The goal of this section 
is to instantiate the high level specification of the introduced protocol into 
a 
concrete class of protocols that can be used to analyze the security and 
communication complexity trade-off commonly present in memory attestation and 
memory erasure protocols~\cite{SMKK2005,KK2014,PT2010,DKW2011,GW2015}. 
%Moreover, the proposed class of protocols follows the standard design of 
%distance bounding protocols where a series of rapid bit-exchanges is used to 
%determine an upper 
%bound on the distance between the prover and the verifier.
%. 

\subsection{Lookup-based memory erasure protocols}

The instantiation we propose is 
largely inspired by the design of lookup-based distance bounding 
protocols~\cite{HK2005,AT2009,GAA2010,KKBD2011,KA2011,MTT2016,MP2008,TMA2010}, 
and we will use the automata-based 
representation introduced by Mauw et al.~\cite{MTT2016} to describe them. The 
main feature of this type of protocols is that, in order to obtain tight values 
on the round-trip-time calculation, they use simple lookup operations during 
the time-measurement phase.

An automaton, i.e. a state-labeled Deterministic 
Finite 
Automaton (DFA), is of the form 
$\tuple$,
where $\insym$ is a set of input symbols, $\outsym$ is a set of output 
symbols, $\states$ is a set of states, $\initialstate \in \states$ is the 
initial state, 
$\transfunct\colon \states\times\insym \to \states$ is a transition function, 
and 
$\outfunct\colon \states\to \outsym$ is a 
labeling 
function. 
%Given input and output symbol sets $\insym$ and $\outsym$, 
%respectively, we 
%use 
%$\universe$ to 
%denote the universe of all DFAs over $\insym$ and $\outsym$.

%\begin{definition}[\Dba{}~\cite{MTT2016}] \label{def_db_protocol1} A 
%\emph{lookup-based
%		distance-bounding protocol}, lookup-based protocol for short, with 
%	input set $\insym$ and output set $\outsym$ is a finite non-empty subset of 
%	$\universe$.
%\end{definition}
\begin{example}[Cyclic tree 
automata.]
As a running example, we consider a concrete type of automaton $\tuple$, called 
\emph{cyclic tree 
automaton}, that resembles the tree structure used in~\cite{AT2009}. Cyclic 
tree automata (see Figure~\ref{fig-tree-shaped-automaton}) are 
characterized by 
a depth $\depth$, which determines the 
number of states in $\states$ to be equal to $2^{\depth+1}-1$. The input and 
output symbol sets are binary, i.e. $\insym = \outsym = \{0, 1\}$, and the 
transition function is defined in two steps. First, given the set of states 
$\states = \{\astate_0, \ldots, \astate_{2^{\depth+1}-2}\}$, 
$\transfunct(\astate_i, 0) = \astate_{2i+1}$ and $\transfunct(\astate_i, 1) = 
\astate_{2i+2}$, for every $i \in \{0, \ldots, 2^{\depth}-2\}$. The remaining 
states connect back to $q_1$ and $q_2$ as follows,  $\transfunct(\astate_i, 0) 
= 
\astate_{1}$ and $\transfunct(\astate_i, 1) = 
\astate_{2}$, for every $i \in 
\{2^{\depth}-1, \ldots, 2^{\depth+1}-2\}$. 
%A graphical 
%representation of a cyclic tree automaton with depth $\depth = 2$ is given  in 
%Figure~\ref{fig-tree-shaped-automaton}. 
\end{example}

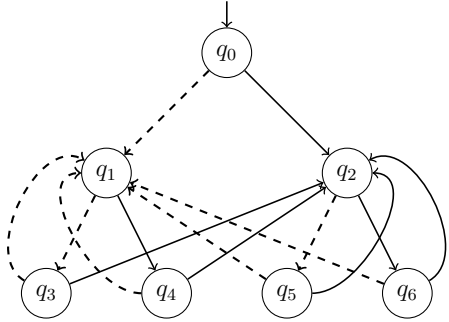
\begin{figure}[htb]
	\centering
\begin{tikzpicture}[scale=0.80, transform shape]
\usetikzlibrary{shapes}
\tikzstyle{myarrow}=[->,thick, dashed]
\tikzstyle{myarrow2}=[->,semithick]

   \node (invisible) at (3,9) {};
   \node (initial) at (3,8)[circle,draw]{$q_0$};
   \node (q1) at (1,6) [circle,draw] {$q_1$};
   \node (q2) at (5,6) [circle,draw] {$q_2$};

   \node (q3) at (0,4) [circle,draw] {$q_3$};
   \node (q4) at (2,4) [circle,draw] {$q_4$};
   \node (q5) at (4,4) [circle,draw] {$q_5$};
   \node (q6) at (6,4) [circle,draw] {$q_6$};

   \draw [myarrow2] (invisible) to node {} (initial);
   \draw [myarrow] (initial) to node [left] [near start]   {} (q1);
   \draw [myarrow2] (initial) to node [right] [near start] {} (q2);

   \draw [myarrow] (q1) to node [left] [near start]   {} (q3);
   \draw [myarrow2] (q1) to node [right] [near start] {} (q4);
   
   \draw [myarrow] (q2) to node [left] [near start]   {} (q5);
   \draw [myarrow2] (q2) to node [right] [near start] {} (q6);

   \draw [myarrow,out=150, in=150] (q3) to node [left] [near start]   {} (q1);
   \draw [myarrow2] (q3) to node [right] [near start] {} (q2);

   \draw [myarrow,out=180, in=180] (q4) to node [left] [near start]   {} (q1);
   \draw [myarrow2] (q4) to node [right] [near start] {} (q2);

   \draw [myarrow] (q5) to node [left] [near start]   {} (q1);
   \draw [myarrow2, out=0, in=0] (q5) to node [right] [near start] {} (q2);

   \draw [myarrow] (q6) to node [left] [near start]   {} (q1);
   \draw [myarrow2, out=30, in=30] (q6) to node [right] [near start] {} (q2);

\end{tikzpicture}
	\caption{A cyclic tree automaton with 
	depth $2$. Dashed and solid edges represent transitions with input symbol 
	$0$ and $1$, respectively.}
	\label{fig-tree-shaped-automaton}
\end{figure}

In a lookup-based protocol, prover and verifier move through a given automaton 
in a synchronous way by feeding the transition function with a sequence of 
challenges sent by the verifier; starting from the initial state. For example, 
in Figure~\ref{fig-tree-shaped-automaton}, if the 
verifier sends challenges $0$, $1$ and $1$, then both prover and verifier are 
meant to follow the path $q_1 q_4 q_2$. Each of those transitions are regarded 
as a \emph{lookup} operation, because the prover's responses are determined by 
the labels of the states in the path. Taking back our running example, the 
prover's response to challenge $0$ is $\outfunct(q_1)$, to a second challenge 
$1$ is $\outfunct(q_4)$, and to a third challenge $1$ is $\outfunct(q_2)$. 

The formalization of the prover-to-verifier interaction described 
above is as follows. Given an automaton $A = \tuple$ and a current state $q \in 
\states$, a 
lookup 
operation 
is regarded as a transition to a new state $q'=\transfunct(q,c)$ where 
$c \in \insym$ is a 
verifier's challenge. The corresponding response for such 
challenge is the output symbol attached to the new state $q'$, i.e.,
$\outfunct(q')$. We use $\statefunction(c_0 \cdots c_i)$ to denote the state 
reached by the sequence of input symbols $c_0, \ldots, c_i$. 
Formally, $\statefunction(c_0 \cdots c_i) = \transfunct(\statefunction(c_0 
\dots c_{i-1}), 
	c_i)$ if $i > 0$, otherwise $\statefunction(c_0) = \transfunct(q_0, c_0)$. 
Similarly, $\outputfunction(c_0 \cdots c_i) = \outfunct(\statefunction(c_0 
\cdots c_i))$ is used to 
denote the output symbol assigned to the state reached by the sequence $c_0 
\cdots c_i$. Finally, the sequence of output symbols resulting from the input 
sequence $c_0 \ldots c_i$ in an automaton $A$ is denoted $\AIn{A}{c_0 \cdots 
c_i}$.

Figure~\ref{fig-prot} depicts our class of lookup-based memory erasure 
protocols. It consists of three phases. An \emph{initial phase} where verifier 
and prover agree on an automaton $\tuple$, such as a cyclic tree automaton. In 
this phase, the prover 
authenticates 
the verifier to prevent unauthorized readers from erasing its memory. After the 
initial phase, the \emph{fast phase} starts, consisting of a series of rapid 
bit exchanges where the prover is challenged to traverse the automaton 
generated during the initialization phase. In the \emph{final phase} the 
verifier takes a decision based on the round-trip-times values and the prover's 
responses obtained during the fast phase. Details on each phase is 
given next.

\begin{figure}\centering
\includegraphics[scale=0.75]{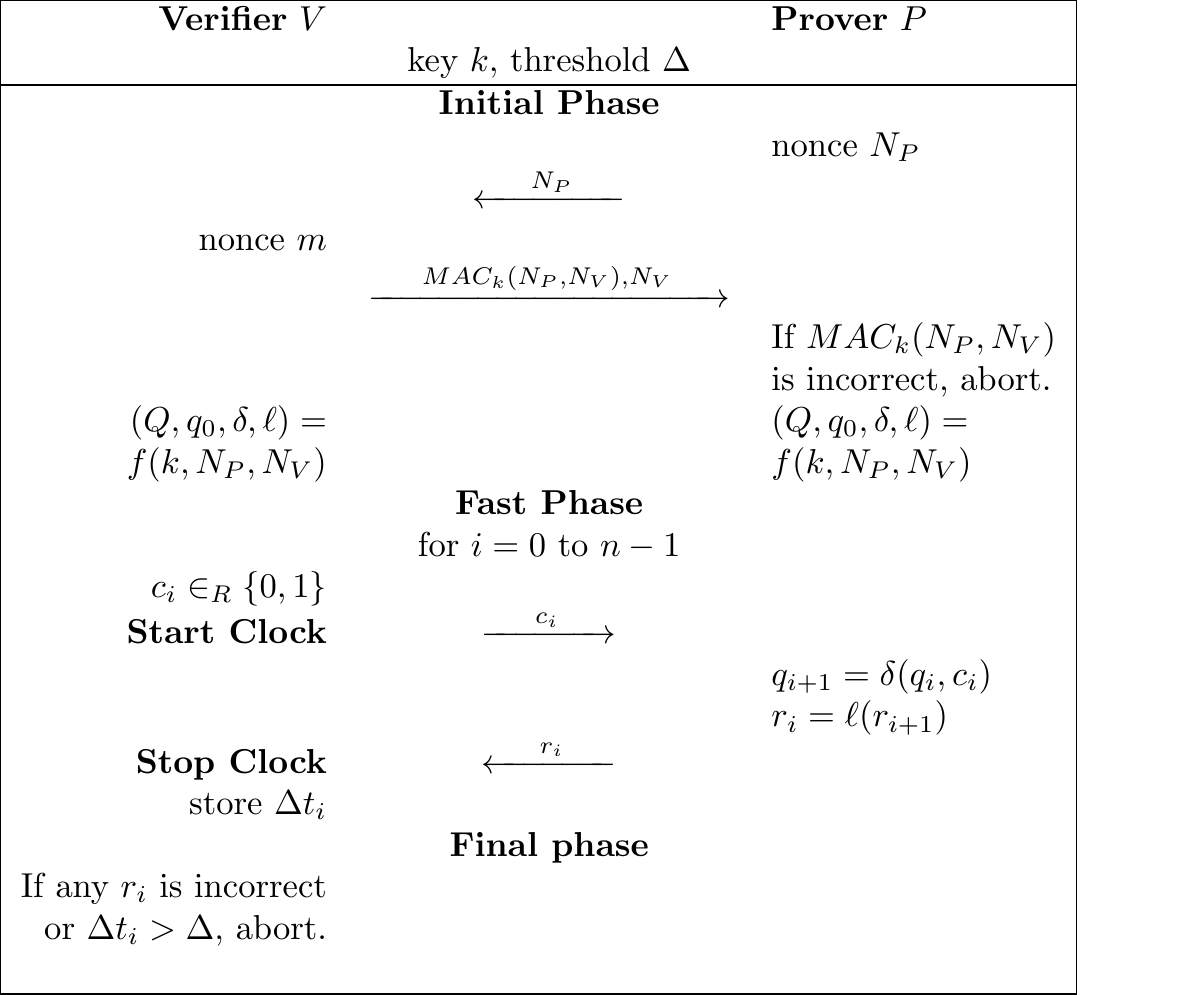}
\caption{The class of lookup-based memory erasure protocol.}
\label{fig-prot}
%\lessspace
\end{figure}

\noindent \emph{Initial phase.} 
%\headingfourthlevel{Initial phase} 
As in the high-level specification depicted in Figure~\ref{fig-prot-symbolic}, 
the first two messages of the protocol are used by the prover to authenticate 
the 
verifier before executing the remainder 
of the protocol, i.e. before erasing its memory. If this authentication step is 
successful, both prover and verifier 
build an automaton $\tuple$ 
based on the output of a pseudo random function $\pseudofunc(.)$ seeded with 
the triple $(k, \nonceP, \nonceV)$. Detail on how such automaton can be built 
based on the 
output of a pseudo-random function can be found 
in~\cite{AT2009} and~\cite{MTT2016}. Here we abstract away from those 
details and consider the output of the initial phase to be a randomly chosen 
automaton from a set of automata. Formally, let $\universe$ be the universe of 
automata 
with input and output symbol 
set $\insym$ and $\outsym$, respectively. Given a lookup-based memory erasure 
protocol $\protocol$, we use $\protocolIni \subseteq 
\universe$ to denote all possible automata that can result from 
the initial phase in $\protocol$.

\noindent \emph{Fast phase.} 
%\headingfourthlevel{Fast phase} 
The fast phase of the protocol $\protocol$ starts 
right after agreeing on a random automaton $\tuple$ from $\protocolIni$. It 
consists of 
the following 
$\rounds$ rounds. For $i = 0$ to $i = \rounds-1$, the verifier picks a random 
bit 
$c_i 
\in_{R} \{0, 1\}$ 
and sends it to 
the prover. Upon reception of $c_i$, the prover applies the transition function 
$q_{i+1} = \transfunct(q_i, c_i)$ and returns the label $\outfunct(q_{i+1})$ of 
the state $q_{i+1}$, with $q_0$ being the initial state of the automaton. The 
verifier stops the timer immediately after 
receiving the prover's response and calculates the round-trip-time $\Delta 
t_i$. 

\noindent \emph{Final phase.}
%\headingfourthlevel{Final phase} 
At the end of the fast phase the verifier 
checks 
that all round-trip-times are below the pre-defined time threshold $\Delta$. 
The verifier also checks that all responses are correct by traversing the 
automaton with its own challenges $c_0, \ldots, c_{n-1}$. If either of those 
verification steps fails, the verifier considers the protocol unsuccessful.

Our main claim here is that the protocol in Figure~\ref{fig-prot-symbolic} is 
an accurate abstraction of the lookup-based memory erasure protocol in 
Figure~\ref{fig-prot}, provided that guessing 
the automaton used during an honest prover-to-verifier execution is 
unfeasible for an attacker. Although this does not necessarily prevent 
probabilistic attacks, as we show 
next, it asserts that lookup-based memory erasure 
protocols contain no logical flaws. That is to say, the introduced lookup-based 
memory 
erasure protocol resists man-in-the-middle attackers as defined in 
Sections~\ref{sec-formal-model} and~\ref{sec-memory-erasure-def}.

\subsection{Security and communication complexity trade-off}
%\subsection{Performance and security trade-off}

The security  and communication complexity trade-off stems from the 
fact that the portion of the prover's memory that can be attested is 
proportional to the communication complexity of the 
protocol~\cite{SMKK2005,KK2014,PT2010,DKW2011,GW2015}. In lookup-based 
memory erasure protocols, we define communication complexity as the number of 
rounds $\rounds$ used during the fast phase, and memory to be attested as 
the size of the automaton agreed upon during the initial phase. 

We consider an implementation-independent notion of size for 
automata, which corresponds to the number of states of an automaton. Formally, 
given $A = \tuple$ we consider the function $\size{.}$ defined by 
$\size{A} = |\states|$. 

In 
an honest session between prover and verifier, the outcome of the initial 
phase is an automaton randomly chosen from the set $\protocolIni$. In this 
case, the memory required on the prover's side to execute the protocol is at 
least the size of the automaton agreed upon with the verifier. However, a 
fraudulent prover may use a smaller automaton with 
the intention of successfully passing the memory erasure protocol without fully 
erasing its memory.

\begin{definition}[Fraudulent prover]
Given a protocol $\protocol$, a \emph{fraudulent prover} is defined by a 
function $\fraudprover: \protocolIni \rightarrow \universe$ such that for every 
$A \in \protocolIni$, $\size{A} 
\geq \size{\fraudprover(A)}$.

The \emph{probability of success} of a fraudulent prover is calculated by, 
given a 
random automaton $A = \tuple \in_{R} 
\protocolIni$ and random sequence $c_0 \cdots c_{n-1} \in_{R} \insym^{n}$,  
\[
\Pr(\AIn{A}{c_0 \cdots c_{n-1}} = \AIn{\fraudprover(A)}{c_0 \cdots 
c_{n-1}})\text{,}
\]

The \emph{space saving} of a fraudulent prover is given by the 
formula, 

\[
1 - \frac{\sum_{A \in \protocolIni}\left( \size{A} - \size{\fraudprover(A)} 
\right)}{\sum_{A \in 
\protocolIni}\size{A}} \text{,}
\]
\end{definition}

%We allow fraudulent provers to deviate from the 
%protocol specification during the initial phase, but not during the fast 
%phase. 
Maximizing both probability of success and space saving is 
unattainable. 
The smallest automaton a fraudulent prover can use
consists of a single state with two self-transitions, one with $0$ and another 
with $1$. But, in this case its probability of success becomes $1/2^n$, where 
$n$ 
is the number of rounds during the fast phase. 
Thus fraudulent provers will aim at striking a good trade-off between 
probability of success and space savings.

In general, we are interested on an optimal fraudulent prover that achieves the 
maximum probability of success restricted to a given size for the automata. 
This might be achieved by using automata minimization techniques, such 
as~\cite{ZMD2005}, where sub-automata that repeats often are assigned a single 
state. Although this is a promising research direction, we focus in this 
article on a simpler fraudulent strategy that consists of 
ignoring portions of the automaton in order to meet a given memory 
requirement. The problem of determining and analyzing optimal fraudulent 
provers is thus left for future work. 

\subsection{Analysis of a protocol based on cyclic tree automata}

We deliver a concrete trade-off analysis by considering a lookup-based memory 
erasure protocol that only utilizes cyclic tree automata of a 
given depth, called \emph{tree-based memory erasure} protocol. That is, given 
the universe of cyclic tree automata with depth $\depth$, denoted 
$\treeuniverse$, we 
define the 
\emph{tree-based memory erasure} protocol to be a lookup-based 
protocol with $\protocol_{ini} = \treeuniverse$. 
We also consider that a fraudulent prover can remove a subtree 
from a cyclic tree automaton, with the idea of leaving room to the malicious 
software to persist in memory.
Formally, given a cyclic tree automaton $A = 
\tuple$ of depth $d$, denoted $\treeauto$, 
the fraudulent prover chooses a state $q_i \in \states$ and 
disconnects it from the tree as follows. For every $q_j \in 
\states$ such that $\transfunct(q_j, b) = q_i$ with $b \in \{0, 1\}$, 
$\transfunct(q_j, b)$ is set to be equal to $\transfunct(q_j, \neg b)$. 
The resulting set of disconnected states $S_{q_i}$ is inductively defined by 
$q_i \in S_{q_i}$ and $q_x \in S(q_i) \iff \exists q_y \in\states \colon y = 
2x-2 \vee y 
= 2x-1$. States in $S_{q_i}$ are consequently removed from $\states$. We use 
$A_{q_i}$ to denote the resulting automaton.

\begin{theorem}
Let $d$ be a depth value and $\treeuniverse$ the universe of cyclic tree 
automata with state set $\states = \{\astate_0, \ldots, 
\astate_{2^{\depth+1}-2}\}$. Given a state 
$q_i \in \states$ with $i > 0$, let 
$\prover_{q_i}$ be a fraudulent prover defined by $\fraudprover(A) = A_{q_i}$ 
for every $A \in \treeuniverse$. If $n = d \times \repetitions$ for some 
positive integer $x$, then for a 
random automaton $A = \tuple \in_{R} 
\treeuniverse$ and random sequence $\challenge \in_{R} \insym^{n}$,  

\begin{align*}
& \Pr(\AIn{A}{\challenge} = \AIn{\fraudprover(A)}{\challenge}) =  
\left( 1- 
\frac{1}{2^{d_i}} + \frac{1}{2^{n+1}} \right)^x
\text{,}
\end{align*}

\end{theorem}

\begin{proof}
Let $\challenge = c_0 \cdots c_{n-1}$. 
We use $d_{A}(q_i, q_j)$ to denote the distance of a shortest path between 
states $q_i$ and $q_j$ in the automaton $A$. Let $d_i = d_A(q_0, q_i)$ and 
$\tilde{c}_0 \cdots 
\tilde{c}_{d_i-1}$ be the sequence of input symbols such that 
$\outputfunction(\tilde{c}_0 \cdots 
\tilde{c}_{d_i-1}) = q_i$. 

Assume $n = d$. If $c_0 \cdots c_{d_i-1} \neq \tilde{c}_0 \cdots 
\tilde{c}_{d_i-1}$, then $\AIn{A}{c_0 \cdots c_{n-1}} = 
\AIn{\fraudprover(A)}{c_0 \cdots 
c_{n-1}}$ given that no state in $A$ that has been removed in $c(A)$ is used. 
Otherwise, 
If $c_0 \cdots c_{d_i-1} = \tilde{c}_0 \cdots \tilde{c}_{d_i-1}$, $\AIn{A}{c_0 
\cdots c_{n-1}} = \AIn{\fraudprover(A)}{c_0 \cdots c_{n-1}}$ with probability 
$\frac{1}{2^{n-d_i+1}}$, given that $c_{0} \cdots c_{n-1}$ and $A$ are 
randomly chosen. Because the probability of $c_0 \cdots c_{d_i-1} = \tilde{c}_0 
\cdots \tilde{c}_{d_i-1}$ is equal to $\frac{1}{2^{d_i}}$, we obtain an overall 
probability of 

\begin{align*}
\Pr(\AIn{A}{\challenge} = \AIn{\fraudprover(A)}{\challenge}) & = 1- 
\frac{1}{2^{d_i}} + \frac{1}{2^{d_i}}\times\frac{1}{2^{n-d_i+1}} \\
& = 1- 
\frac{1}{2^{d_i}} + \frac{1}{2^{n+1}}\text{,}
\end{align*}

For the general case where $n = \depth \times \repetitions$, we 
notice that, by construction of the cyclic tree automaton, the following 
property holds,

\begin{align*}
\AIn{A}{c_0 \cdots c_{n-1}} = & \AIn{A}{c_0 \cdots c_{d-1}} \concat 
\AIn{A}{c_d 
\cdots c_{2d-1}} \concat \\
& \cdots \concat \AIn{A}{c_{n-d} \cdots c_{n-1}} \text{,}
\end{align*}

The same property holds for the modified automaton $\fraudprover(A)$. 
Therefore, 

\begin{align*}
& \Pr(\AIn{A}{c_0 \cdots c_{n-1}} = \AIn{\fraudprover(A)}{c_0 \cdots 
c_{n-1}}) = \\
& \hspace{0.1cm} \prod_{j = 1}^{j = x} \Pr(\AIn{A}{c_{(j-1)\depth} \cdots 
c_{jd-1}} 
= 
\AIn{\fraudprover(A)}{c_{(j-1)\depth} \cdots 
c_{jd-1}}) = \\
& \hspace{0.1cm} \left( 1- 
\frac{1}{2^{d_i}} + \frac{1}{2^{n+1}} \right)^x
\text{,}
\end{align*}

\end{proof}

To illustrate the security and communication complexity trade-off, we consider 
a 
cyclic tree automaton of depth $\depth = 12$, which gives $2^{13}-1$ states. We 
claim that an automaton of this size requires at least $1$Kb of memory based on 
a rough conversion of $1$bit per state. Figure~\ref{fig-tradeoff} depicts, for 
different values of the number of 
rounds $\rounds$, 
the space saving 
and success probability achieved by an attacker that uses the strategy of 
removing a full subtree from the automata. Larger space saving is achieved by 
disconnecting states closer to the root state $q_0$.

\begin{figure}\centering
\includegraphics[scale=0.65]{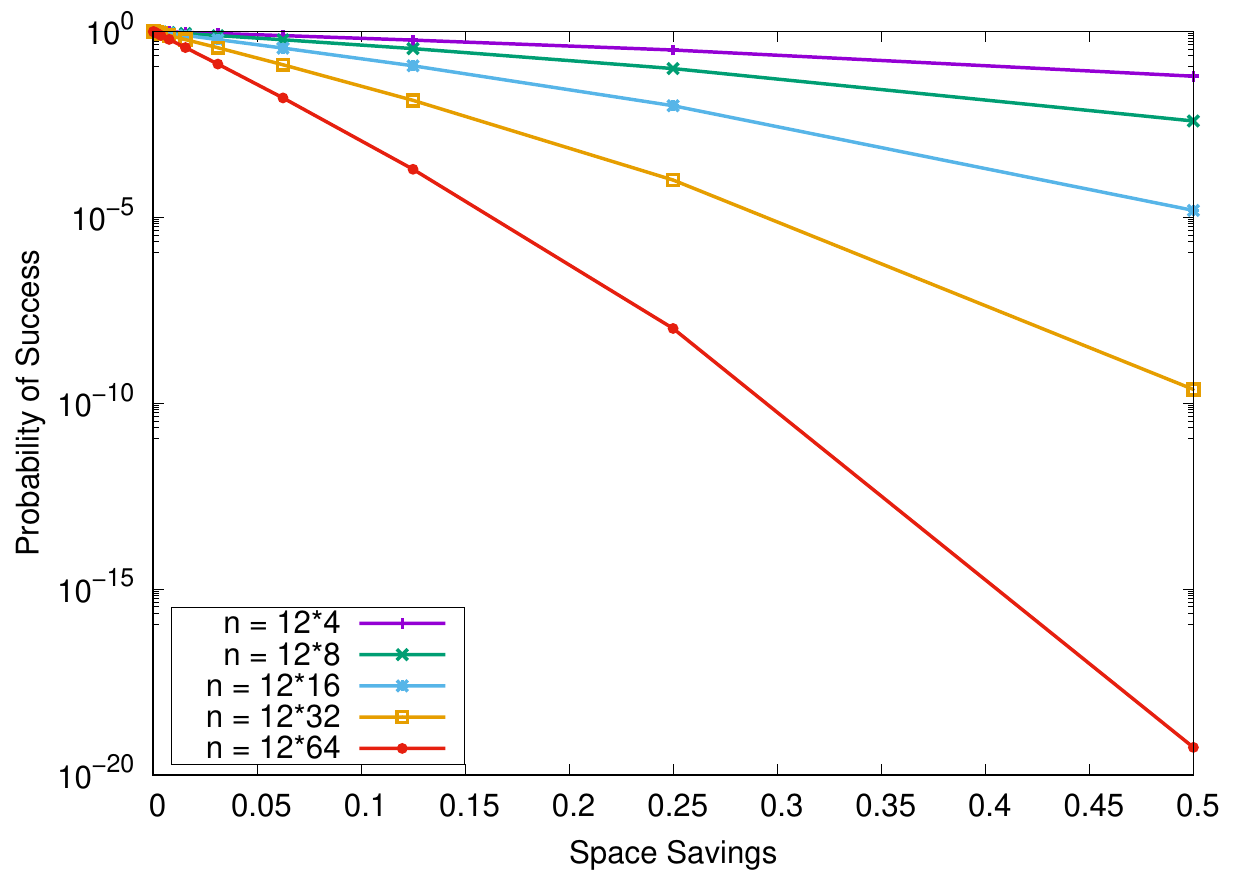}
\caption{Trade-off between communication complexity and security. The y-axis is 
in logarithmic scale.}
\label{fig-tradeoff}
%\lessspace
\end{figure}

The two expected trade-off can be observed in  Figure~\ref{fig-tradeoff}. On 
the one hand, the larger the space saving the smaller the probability of 
success of the considered strategy. On the other hand, the security of the 
protocol increases with the number of rounds. For space savings of around $10 
\%$, the fraudulent prover succeeds with high probability, unless $n$ is 
sufficiently large. For example, $n = 12 \times 4$ gives a probability of 
success of $0.58$, while $n = 12 \times 64$ gives $1.96 \times 10^{-4}$. In 
comparison to distance bounding protocols, where $48$ rounds are regarded as a 
good balance between security and communication complexity, lookup-based memory 
erasure protocols seem to require significantly more message exchanges. We 
remark, however, that this problem is inherent to most remote memory 
attestation 
and memory erasure procedures.

\section{Conclusion}\label{sec-conclusions}

In this article we addressed the problem of formal verification of memory 
erasure protocols. We used a symbolic model by Basin et al.~\cite{BCSS2009} to 
provide the first definition of secure memory erasure that can be used for 
formal reasoning, and proved that no protocol can meet such property against 
the standard Dolev-Yao adversary. This motivated the formalization of a 
slightly weaker 
attacker, called a distant attacker, which is a Dolev-Yao adversary restricted 
to a given distance 
threshold on their interaction with honest participants. Our main result 
consists of the first memory erasure protocol that resists man-in-the-middle 
attacks, which we proved contains no logical flaws based on the protocol 
verification tool 
\tamarin{}~\cite{tamarin13} and recent results on causality-based 
characterization 
of distance bounding protocols~\cite{causalDB18}. Because the considered 
security model cannot reason about message size, we instantiated the introduced 
protocol using standard cryptographic notation for distance bounding protocols 
and analyzed the resulting security and communication complexity trade-off.

%We believe that the limitations of the work presented herein lead to promising 
%research 
%directions. For example, determining and implementing optimal fraudulent 
%strategies against lookup-based secure memory erasure protocols 
%as well as finding a class of automata that optimizes the security 
%communication trade-off problem remain open problems. Also, the need of 
%making explicit the notion of size in a formal definition of secure memory 
%erasure within a symbolic security model is unclear for the authors, hence 
%calling for further research. Last, but not least, addressing the open problem 
%of 
%designing a memory attestation protocol that resists man-in-the-middle 
%attacks, 
%as we do in this paper with the topic of memory erasure, seems an interesting 
%research challenge by itself. 
%

%\bibliographystyle{plain}
%\bibliography{bibliography/bibliography}

\end{document}